\documentclass[11pt]{article}

\usepackage[margin=1in]{geometry}
\usepackage{amssymb}
\usepackage{amsmath}
\usepackage{amsthm}
\usepackage{mathtools}
\usepackage[inline]{enumitem}
\usepackage[textfont=it,tableposition=above,font=small]{caption}
\usepackage[colorlinks,citecolor=blue,urlcolor=blue]{hyperref}
\usepackage[symbol]{footmisc}
\usepackage{xcolor}
\usepackage{arydshln}
\usepackage{multirow}
\usepackage{lmodern}
\usepackage{natbib}

\newtheorem{theorem}{Theorem}
\newtheorem{lemma}[theorem]{Lemma}

\theoremstyle{definition}
\newtheorem*{remark}{Remark}
\newtheorem*{definition}{Definition}

\newcommand{\RPD}{\textsf{\textbf{RPD}}}
\newcommand{\RHD}{\textsf{\textbf{RHD}}}
\newcommand{\FD}{\textsf{\textbf{FD}}}
\newcommand{\ID}{\textsf{\textbf{ID}}}

\newcommand{\VV}{\mathcal{V}}

\newcommand{\HH}{\mathbb{H}}
\newcommand{\R}{\mathbb{R}}
\newcommand{\N}{\mathbb{N}}

\newcommand{\argmax}{\operatornamewithlimits{argmax}}

\newcommand{\pr}{\mathbb{P}}
\newcommand{\E}{\mathbb{E}}

\newcommand{\SD}{\mathsf{SD}}
\newcommand{\MAD}{\mathsf{MAD}}
\newcommand{\med}{\mathsf{med}}

\newcommand{\set}[1]{\left\lbrace #1 \right\rbrace}
\renewcommand{\SS}{\mathbb{S}}
\newcommand{\BB}{\mathbb{B}}
\renewcommand{\P}[1]{\mathcal{P}({#1})}
\newcommand{\dd}{\,\mathrm{d}}

\newcommand{\abs}[1]{\left|{#1}\right|}
\newcommand{\inner}[1]{\left\langle#1\right\rangle}
\newcommand{\norm}[1]{\left\lVert#1\right\rVert}
\newcommand{\Clo}[1]{\operatorname{cl}\left({#1}\right)}

\newcommand{\spn}[1]{\operatorname{span}\left(#1\right)}
\newcommand{\st}{\,\colon\,}
\newcommand{\eqdis}{\overset{d}{=}}
\newcommand{\tr}{^\top}

\newcommand{\distto}{\overset{d}{\longrightarrow}}
\newcommand{\xdistto}[1][]{\xrightarrow[\scalebox{0.7}{$#1$}]{d}}

\newcommand{\weakto}{\rightharpoonup}
\newcommand{\xweakto}[1][]{\xrightharpoonup[\scalebox{0.7}{$#1$}]{}}

\newcommand{\asto}{\xrightarrow[]{\text{a.s.}}}
\newcommand{\xasto}[1][]{\xrightarrow[\scalebox{0.7}{$#1$}]{\text{a.s.}}}

\newcommand{\xto}[1][]{\xrightarrow[\scalebox{0.7}{$#1$}]{}}

\title{Projection depth for functional data: Theoretical properties\thanks{
This work was supported by the Czech Science Foundation (project no.~24-10822S)
and the ERC CZ grant LL2407 of the Ministry of Education, Youth and Sport of the Czech Republic.
}}

\author{
Filip Bo\v{c}inec\\
Faculty of Mathematics and Physics, Charles University, Prague, Czech Republic\\
\texttt{bocinec@karlin.mff.cuni.cz}
\and
Stanislav Nagy\\
Faculty of Mathematics and Physics, Charles University, Prague, Czech Republic\\
\texttt{nagy@karlin.mff.cuni.cz}
\and
Hyemin Yeon\\
Department of Mathematical Sciences, Kent State University, Kent, USA\\
\texttt{hyeon1@kent.edu}
}

\date{}

\begin{document}
\maketitle

\begin{abstract}
We introduce a novel projection depth for data lying in a general Hilbert space,
called the regularized projection depth,
with a focus on functional data.
By regularizing projection directions,
the proposed depth does not suffer from the degeneracy issue that may arise when the classical projection depth is naively defined on an infinite-dimensional space. 
Compared to existing functional depth notions, 
the regularized projection depth has several advantages: 
(i) it requires no moment assumptions on the underlying distribution, 
(ii) it satisfies many desirable depth properties including invariance, monotonicity, and vanishing at infinity,  
(iii) its sample version uniformly converges under mild conditions,
and (iv) it generates a highly robust median.
Furthermore, the proposed depth is statistically useful as 
it (v) does not produce ties in the induced ranks 
and (vi) effectively detects shape outlying functions. 
This paper focuses mainly on the theoretical properties of the regularized projection depth.
\end{abstract}

\bigskip
\noindent\textbf{Keywords:}
functional data analysis; statistical depth; projection depth; robust statistics

\medskip
\noindent\textbf{MSC 2020:}
62G05; 62H99; 62R10

\section{Introduction}

Many tools in both theoretical and applied statistics require ordering data points according to their “typicality” or “centrality” within a dataset. 
In the univariate setting, this task is naturally solved by sorting the observed values.
The problem becomes more involved 
as data structure exhibits more complexity.
For instance, for multivariate, functional, manifold, or object data,
it is not straightforward to define canonical ordering.
To address this challenge, the concept of statistical depth was introduced. 
A statistical depth function assigns a non-negative number to each point in the underlying space where the observed random elements lie.
The corresponding depth value quantifies its centrality with respect to (w.r.t.) a given distribution or data cloud. 
Points with high depth are considered central, while those with low depth lie near the boundary and may be considered potential outliers.
Among the many applications of depth functions are robust estimation, outlier detection, hypothesis testing, and data visualization; see \citet{Liu_etal1999, Zuo_Serfling2000, Mosler_Mozharovskyi2022} and the references therein. 

Numerous multivariate depth functions have been introduced in the literature. 
The earliest and the most thoroughly studied is the halfspace (or Tukey) depth~\citep{Tukey1975}.
It has inspired a wide range of other depth notions such as the simplicial depth~\citep{Liu1990}, the spatial depth~\citep{Chaudhuri1996}, and the zonoid depth~\citep{Mosler2002}. 
A prominent notion is the \emph{projection depth}~\citep{Zuo2003}. 
It is based on the concept of projection pursuit~\citep{Huber1985} and the Stahel-Donoho outlyingness~\citep{Stahel1981, Donoho1982}. 

To describe the multivariate projection depth, 
let $\HH = \R^d$ be the $d$-dimensional Euclidean space with the usual inner product $\left\langle \cdot, \cdot \right\rangle$ and the associated Euclidean norm $\left\Vert \cdot \right\Vert$. 
For a point $x\in \HH$ and a unit vector $v \in \SS = \set{y\in\HH \st \norm{y}=1}$, the outlyingness of $x$ w.r.t. an $\HH$-valued random variable $X \sim P_X$ in direction $v$ is
\begin{align}\label{eqOut}
	O_v(x;P_X) = \frac{\abs{\inner{x, v} - \med[\inner{X,v}]}}{\MAD[\inner{X,v}]}.
\end{align}
Here, $\med[U]$ and $\MAD[U]$ stand for the median and the median absolute deviation (MAD) of the real-valued random variable $U$.\footnote{We define the median of a real-valued random variable~$U$, denoted by~$\med[U]$, as the midpoint of all values~$m \in \R$ satisfying $\min\left\{ \pr(U \leq m), \pr(U \geq m) \right\} \geq \tfrac{1}{2}$. 
The median absolute deviation $\MAD[U]$ is then defined as $\med[\abs{U - \med[U]}]$.}
The projection depth at $x \in \HH$ w.r.t. $P_X$ is then defined as
\begin{align}\label{eqFPD}
	D(x;P_X)= \inf_{v \in \SS} \left( 1 + O_v(x; P_X) \right)^{-1}=\left( 1 + \sup_{v \in \SS}O_v(x; P_X) \right)^{-1},
\end{align}
that is, we take the largest outlyingness over all directions and transform it so that points that are less outlying receive higher depth values.
The projection depth \eqref{eqFPD} on $\R^d$ satisfies a range of desirable properties related to invariance, symmetry, and consistency. \citet[Table~2]{Mosler_Mozharovskyi2022} provide an overview of the relevant properties of the projection depth in comparison with other depth concepts. \citet{Liu_Zuo2014b} examine computational aspects of this depth.

Beyond multivariate data, 
modern statistical applications often involve datasets in which each observation is a function such as a curve or a surface.
As a result, the framework of functional data analysis has become increasingly relevant~\citep{Ramsay_Silverman2005, Ferraty_Vieu2006, Hsing_Eubank2015}.
As in the multivariate setting, it is important to assess how central or outlying a given function is within a sample.
This need has led to the development of \emph{functional depth}, which has been an active research topic for over two decades.
Beginning with the work of~\citet{Fraiman_Muniz2001}, several notions of functional depth have been proposed and studied,
including those by~\citet{Lopez_Romo2009, Mosler2013, Chakraborty_Chaudhuri2014, Nagy_etal2016, Narisetty_Nair2016, BBC20, YDL25RHD}.

We now denote by $\HH$ a (typically infinite-dimensional) Hilbert space equipped with the inner product $\inner{\cdot, \cdot}$ and the associated norm $\norm{\cdot}$. 
A natural approach is to generalize a depth function on the Euclidean space $\R^d$---such as halfspace depth or projection depth~\eqref{eqFPD}---to this more general functional setting by simply replacing $\R^d$ in~\eqref{eqFPD} by a Hilbert space $\HH$ and by considering the unit sphere $\SS = \{y \in \HH \colon \norm{y} = 1\}$ in $\HH$. 
However, it is well-known that a naive extension of multivariate depth functions to functional data may suffer from the so-called \emph{degeneracy}, meaning that the depth is zero almost everywhere. This was shown by \citet[Theorem~3]{Dutta_etal2011} and \citet{Kuelbs_Zinn2013} for halfspace depth, and by \citet{Chakraborty_Chaudhuri2014} for both halfspace and projection depths.\footnote{\citet[Theorem~2.2]{Chakraborty_Chaudhuri2014} show that the usual projection depth degenerates in the space of continuous functions on $[0,1]$ under a Gaussian underlying distribution.
Our Theorem~\ref{thmFPDdege} shows that this is true in any infinite-dimensional Hilbert space.}
As we show in the following result, the naive extension of the projection depth degenerates for any Gaussian distribution in an infinite-dimensional~$\HH$.
The proof of Theorem~\ref{thmFPDdege} is in the Appendix.

\begin{theorem}[Degeneracy of projection depth] \label{thmFPDdege}
Suppose $\dim(\HH)=\infty$ and $X \sim P_X$ is a Gaussian random element on $\HH$ with positive definite covariance operator. Then, for an independent copy $X_0$ of $X$, we have $D(X_0; P_X) = 0$ almost surely.
\end{theorem}

As a consequence, the naive extension of projection depth to a functional setup does not provide any sensible ranking of data, already in the canonical case of a Gaussian random variable. 
The core issue in infinite-dimensional spaces $\HH$ is that the set of directions $\SS$ is too large. 
Due to the infinite dimensionality of $\HH$ and the tightness of the distribution of $X \sim P_X$ living in $\HH$, it is possible to find a direction $v \in \SS$ such that  $\MAD[\inner{X, v}]$ is arbitrarily small.
Consequently, the outlyingness~\eqref{eqOut} in such a direction becomes arbitrarily large, which in turn causes the depth to degenerate to zero.  

To overcome the degeneracy identified in Theorem~\ref{thmFPDdege}, 
this paper introduces the \emph{regularized projection depth} (\RPD{}), a novel statistical depth for random elements that take values in a (typically infinite-dimensional) Hilbert space, referred to as \emph{Hilbertian elements} \citep{JP20}. 
The key idea is to regularize the MADs of projections in order
to prevent them from taking extremely small values.
In effect, we exclude the directions along which projected data are highly concentrated around their medians.
This restriction prevents the outlyingness defined in \eqref{eqOut} from diverging to infinity and ensures that the corresponding projection depth \eqref{eqFPD} remains non-degenerate, as shown in Theorem~\ref{thm: nonDeg} below.
We establish several desirable theoretical properties of the~\RPD{}, including invariance, maximality at center, monotonicity, and continuity, under appropriate modification to the postulates by \cite{Zuo_Serfling2000, Nieto_Battey2016, Gijbels_Nagy2017}. We also show that the sample version of~\RPD{} can be consistently estimated, and that the induced median is robust.

\RPD{} is obtained using a regularization approach, that is, by restricting the set of directions $\SS$ in~\eqref{eqFPD}. This is motivated by the recently introduced regularized halfspace depth \citep{YDL25RHD}, which applies a related idea to construct a halfspace depth that is suitable for functional data. Although both these depths use regularization, the two approaches are fundamentally different.
The definition of regularized halfspace depth depends on the covariance operator and involves the reproducing kernel Hilbert space norm. As a result, that depth inherits technical limitations and requires strong assumptions. In contrast, \RPD{} does not require the existence of any moment of the underlying distribution, which makes \RPD{} suitable even for heavy-tailed data.
Moreover, \RPD{} exhibits greater robustness as it is defined in a fully nonparametric manner. Since \RPD{} does not rely on the covariance operator, its sample version also avoids cumbersome dependencies on principal components and the associated choice of truncation parameters as encountered in \citet{YDL25RHD}.
In summary, \RPD{} is conceptually simpler and more broadly applicable than the functional depth from~\cite{YDL25RHD}; additional advantages of \RPD{} will appear in its performance in practical tasks.

\subsection*{Contributions and structure of the paper}

In Section~\ref{sec2}, we introduce \RPD{}, and formally establish its non-degeneracy. Section~\ref{sec3} establishes the main theoretical properties of \RPD{} within a general Hilbert space framework, including its sample version consistency, and robustness of the induced median. Section~\ref{sec4} discusses some practical aspects of the~\RPD{}. We apply it to outlier detection and showcase its strengths with comparison to the classical functional depths. Proofs of all theoretical results are given in the Appendix.

\subsection*{Notations}

We consider a separable (typically infinite-dimensional) Hilbert space $\HH$ with inner product $\inner{\cdot, \cdot}$ and the induced norm $\norm{\cdot}$. Let $\SS = \set{v \in \HH \st \norm{v}=1}$ denote the unit sphere in $\HH$. For a set $A \subseteq \HH$, we use $\Clo{A}$ and $\spn{A}$ to denote its closure and linear span, respectively. The notation $x_k \to x$ ($x_k \weakto x$) stands for (weak) convergence in $\HH$ as $k \to \infty$~\citep[Definition~3.2.10]{Hsing_Eubank2015}.
All random variables are defined on a common probability space $\left( \Omega, \mathcal F, \pr \right)$. Throughout the paper, “a.s.” stands for “almost surely.” The set of all Borel probability measures on $\HH$ is denoted by $\mathcal{P}(\HH)$. For a random variable $X \sim P_X \in \P{\HH}$ and a measurable map $\varphi$ from $\HH$ to another Hilbert space, we write $P_{\varphi(X)}$ for the distribution of $\varphi(X)$.
The symbol \scalebox{0.8}{$\eqdis$} denotes equality in distribution. 
For real-valued random variables $\{Z_k\}_{k=1}^\infty$ and $Z$,
the notation 
$Z_k \distto Z$ ($Z_k \asto Z$) denotes convergence in distribution (a.s.) of $\{Z_k\}_{k=1}^\infty$ to $Z$ as $k \to \infty$.

\section{Regularized projection depth} \label{sec2}

Inspired by the approach of~\cite{YDL25RHD},
we aim to incorporate regularization into the definition of the classical projection depth
in order to avoid the degeneracy described in Theorem~\ref{thmFPDdege}.
To do so, we restrict the set of directions considered in its definition in~\eqref{eqFPD}.

\begin{definition}
For $\beta \geq 0$ and $P_X \in \P{\HH}$, we consider the set of regularized directions
\begin{equation}\label{eqRegDir}
    \VV_{\beta} = \VV_{\beta}(P_X) = \set{v\in\SS\st \MAD[\inner{X,v}]\geq \beta}.
\end{equation}
The \emph{regularized projection depth (\RPD{})} of $x \in \HH$ w.r.t. $P_X$ is defined as
	\begin{align}
		D_{\beta}(x;P_X)
		= \inf_{v \in \VV_{\beta}} \left( 1 + O_v(x; P_X) \right)^{-1}, \label{eqRPD}
	\end{align}
where the projection outlyingness $O_v$ is given in~\eqref{eqOut}.
\end{definition}

\RPD{} with $\beta=0$ is the classical projection depth. 
In the regularization step, we consider only those $v \in \SS$ that satisfy $\MAD[\inner{X,v}] \geq \beta$, that is, directions $v$ onto which projections of $X \sim P_X$ are sufficiently dispersed.
For $\beta>0$, this ensures that \RPD{} remains non-degenerate, as shown in Theorem~\ref{thm: nonDeg} below.

Note that large values of $\beta$ lead to stronger regularization, 
i.e., a smaller set $\VV_{\beta}$.
In the extreme case, the direction set $\VV_{\beta}$ may be empty, which is undesirable.
Therefore, throughout most of the paper, we assume $\VV_{\beta}\neq\emptyset$.
Reasonable choices of $\beta$ are discussed in Theorem~\ref{thm: Vprop}.
In practice (see Section~\ref{sec4}), we set $\beta$ to be the $u$-quantile of $\MAD[\inner{X,V}]$,
where $V$ is independent of $X$ and follows a distribution supported in~$\SS$. 
Here $u \in [0,1)$, where smaller values of $u$ correspond to weaker regularization. 
This construction ensures $\VV_{\beta}\neq\emptyset$.

The following theorem shows that $\beta>0$ guarantees non-degeneracy of \RPD{}. In contrast to the existing related results \citep{Kuelbs_Zinn2013, Dutta_etal2011, Chakraborty_Chaudhuri2014, YDL25RHD}, the non-degeneracy result in Theorem~\ref{thm: nonDeg} holds without stringent conditions on the distribution of $X$.

\begin{theorem}[Non-degeneracy of~\RPD{}]\label{thm: nonDeg}
    Let $X \sim P_X\in \P{\HH}$ and $\beta>0$ be such that $\VV_{\beta}\neq\emptyset$.
    Then $D_{\beta}(x;P_X)$ is positive for all $x\in\HH$.
\end{theorem}

An important family of well-behaved functional distributions is the \emph{elliptically symmetric distributions}~\citep{Boente_etal2014}. For such distributions, the regularized set of directions~\eqref{eqRegDir} simplifies substantially.
Recall that a random element $X \sim P_X \in \P{\HH}$ is said to be elliptically symmetric if the characteristic function of each projection $\inner{X, v}$ of $X$ is of the form
\begin{equation}\label{eq: ElliptProj}
    e^{i \inner{\mu, v} t} \, \varphi\left(\norm{\Sigma^{1/2} v}^2 t^2\right), \quad t \in \R,
\end{equation}
where $\varphi\colon \R \to \R$, $\mu \in \HH$ is a location parameter, and $\Sigma \colon \HH \to \HH$ is a self-adjoint, positive semi-definite, and compact operator on $\HH$ representing the scatter of $X$.  
Here, for a positive semi-definite operator $A$ on $\HH$, $A^{1/2}$ stands for its square-root operator satisfying $A^{1/2}A^{1/2} = A$.
Gaussian distributions are a canonical example, obtained by choosing $\varphi(t) = \exp(-t/2)$.  
However, the class also includes a broader family of Hilbertian elements of the form $X \eqdis \mu + R G$, where $G$ is a Gaussian random element and $R$ is a non-negative random variable independent of $G$~\cite[cf.][Proposition~2.1]{Boente_etal2014}.

The form of the characteristic function~\eqref{eq: ElliptProj} implies that the standardized projection  
\[
    Y = \frac{\inner{X-\mu,v}}{\norm{\Sigma^{1/2} v}}
\]
of $X$ onto $v \neq 0$
has characteristic function $t \mapsto \varphi(t^2)$, and its distribution is thus independent of $\mu$, $\Sigma$, and $v$. Hence, 
\begin{equation}\label{eq: ElliptMAD}
    \MAD[\inner{X,v}] = \MAD[\inner{X-\mu,v}]=\med[\abs{\inner{X-\mu,v}}] = F_v^{-1}(3/4) = b\,\norm{\Sigma^{1/2} v},
\end{equation}
where $F_v^{-1}$ denotes the quantile function of $\inner{X-\mu,v}$ and $b\geq0$ is the $3/4$-quantile of $Y$. The second and third equality in~\eqref{eq: ElliptMAD} follow from the symmetry of $\inner{X-\mu,v}$.
Note that, without loss of generality, we may assume that 
$b=1$; otherwise, we can reparametrize~\eqref{eq: ElliptProj} with $\Sigma'=b^2\Sigma$.
Thus, in the case of elliptically symmetric distributions, the regularized direction set~\eqref{eqRegDir} can be rewritten as
\begin{equation}\label{eq_DirSet_ECD}
\VV_{\beta}
= \set{v\in\SS\st \norm{\Sigma^{1/2} v} \geq \beta}.
\end{equation}
This simplified form mirrors the structure of the direction set $\mathcal{U}_\alpha = \{ v \in \SS: \norm{\Sigma^{-1/2}v} \leq \alpha \}$ used in the definition of the regularized halfspace depth by~\citet{YDL25RHD}.
Here, we define $\Sigma^{-1/2} \equiv \sum_{j=1}^\infty \sigma_j^{-1/2} (\phi_j \otimes \phi_j)$,
where $(\sigma_j,\phi_j)$ denotes the $j$-th eigenvalue-eigenfunction pair  of $\Sigma$.
Indeed, 
for any $v\in\SS$ with $\norm{\Sigma^{-1/2}v}<\infty$,
we have that
$1=\norm{v}\leq \norm{\Sigma^{-1/2}v}\norm{\Sigma^{1/2}v}$ 
by the Cauchy-Schwarz inequality,
yielding $\mathcal{U}_{\beta^{-1}} \subseteq \VV_\beta$.
Our regularization may thus be viewed as a robust non-parametric extension of the approach of~\citet{YDL25RHD}.

\begin{remark}
    The \RPD{} can be extended to accommodate general location $\mathsf{L}$ and scale $\mathsf{S}$ functionals
    by replacing $\med$ and $\MAD$ with $\mathsf{L}$ and $\mathsf{S}$, respectively;
    see \cite{Zuo2003} for finite-dimensional versions.
    In this general framework, regularization is naturally applied 
    by including directions $v \in \SS$ whose scale $\mathsf{S}(\langle X, v \rangle)$ is at least $\beta > 0$. 
\end{remark}

\section{Theoretical properties} \label{sec3}

Core properties of statistical depth functions have been studied both in finite-dimensional \citep{Liu1990, Zuo_Serfling2000, Mosler_Mozharovskyi2022} and functional settings~\citep{Nieto_Battey2016, Gijbels_Nagy2017}.
We build on this framework to examine to what extent the~\RPD{} retains these desirable features. In Subsection~\ref{sec3.1} we introduce the basic properties of \RPD{}, in Subsection~\ref{sec3.2} we establish a consistency result for \RPD{}, and in Subsection~\ref{sec3.3} we show that the induced median is robust.

\subsection{Basic properties of regularized projection depth}\label{sec3.1}

The \RPD{} relies on the regularized direction set~$\VV_{\beta}$ defined in~\eqref{eqRegDir}.
Weak regularization (small~$\beta$) yields a direction set close to the full unit sphere~$\SS$,
while stronger regularization (large $\beta$) progressively shrinks $\VV_{\beta}$.
We begin by analyzing key properties of this set, which directly influence the behavior of the depth.

\begin{theorem}[Properties of the direction set]\label{thm: Vprop}
    Let $X\sim P_X \in \P{\HH}$. Then
    \begin{enumerate}[label=\upshape{(V\arabic*)}]
    
    \item\label{V1} $\VV_{\beta}$ is antipodally symmetric; that is, $v \in \VV_{\beta}$ if and only if $-v \in \VV_{\beta}$.
    
    \item\label{V2} If $P_X$ has contiguous support,\footnote{We say that a random variable $X \sim P_X \in \P{\HH}$ has contiguous support \citep{Kong_Zuo2010} if for each $v \in \SS$, the support of the real random variable $\inner{X,v}$ is connected.} then both $v \mapsto \med[\inner{X,v}]$ and $v \mapsto \MAD[\inner{X,v}]$ are continuous mappings on $\HH$. Consequently, $\VV_{\beta}$ is closed and complete.

    \item\label{V3} Suppose that $\beta$ is chosen as the $u$-quantile of $\MAD[\inner{X,V}]$, 
    for some $u \in [0,1)$, where $V$ follows a distribution on~$\SS$. Then $\VV_{\beta}$ is non-empty.
    \end{enumerate}
\end{theorem}

Part~\ref{V1} states that the regularized direction set is always symmetric, which is natural since $O_v(\cdot; P_X) = O_{-v}(\cdot; P_X)$. 
Under the mild assumption of contiguous support, the set is also closed and complete; see~\ref{V2}. 
Part~\ref{V3} gives a sufficient condition for the non-emptiness of $\VV_{\beta}$.
This strategy for selecting $\beta$ is used in our practical implementation in Section~\ref{sec4}.
Note that, since the mapping $v \mapsto \MAD[\inner{X,v}]$ can behave quite wildly (in contrast to the analogous function involving the covariance operator), one cannot say much about the exact form of $\VV_\beta$ unless $P_X$ is elliptically symmetric; see~\eqref{eq_DirSet_ECD}.

Further, we investigate properties of the~\RPD{}, including quasi-concavity, invariance, symmetry behavior, and vanishing at infinity.
Recall that the $\tau$-upper level set of the~\RPD{} is defined as
\begin{equation}\label{eq: UpperLevelSet}
    D_{\beta}^\tau(P_X) = \set{x \in \HH \st D_{\beta}(x; P_X) \geq \tau}.
\end{equation}
These level sets describe the central regions of $P_X$ induced by the depth $D_\beta$.

\begin{theorem}[Inherited \RPD{} properties]\label{thm: RPDPropInh}
    Let $X \sim P_X \in \P{\HH}$ and assume that $\beta\geq 0$ is chosen such that~$\VV_{\beta}\neq\emptyset$. Then:
    \begin{enumerate}[label=\upshape{(P\arabic*)}]
        \item(\textbf{Quasi-concavity})\label{P1} $D_{\beta}(\cdot;P_X)$ is quasi-concave; that is, for all $\tau \geq 0$, the $\tau$-upper level set $D_{\beta}^\tau(P_X)$ is convex.
        
        \item(\textbf{Median under symmetry}) Assume that either $P_X$ is halfspace symmetric about $\mu\in\HH$ with contiguous support, or that $P_X$ is centrally symmetric about $\mu$.\footnote{Recall that $X \sim P_X\in\P{\HH}$ is said to be halfspace symmetric about $\mu \in \HH$ if, for all $v \in \SS$, it holds that $\pr(\inner{X, v} \geq \inner{\mu, v}) \geq 1/2$. Furthermore, $X$ is centrally symmetric about $\mu$ if $X - \mu \eqdis \mu - X$.} Then, $D_{\beta}(\cdot;P_X)$ attains its maximum at $\mu$.

        \item(\textbf{Level sets under symmetry}) Assume that $P_X$ is centrally symmetric about $\mu\in\HH$. Then $D_{\beta}(\cdot; P_X)$ is centrally symmetric about $\mu$; that is $D_{\beta}(x;P_X)=D_{\beta}(2\mu-x;P_X)$ for all $x\in\HH$. In particular, each level set $D_\beta^{\tau}(P_X)$ is symmetric around $\mu$.

        \item(\textbf{Monotonicity on rays}) Assume that $z \in \HH$ is a maximizer of $D_{\beta}(\cdot;P_X)$. Then the following statements hold:
        \begin{itemize}
            \item $D_{\beta}(z; P_X) > \inf_{x \in \HH} D_{\beta}(x; P_X) = 0$.
            \item For all $x \in \HH$ and $\delta \in [0, 1]$, we have $D_{\beta}(x; P_X) \leq D_{\beta}(z + \delta (x - z); P_X)$.
        \end{itemize}
    \end{enumerate}
\end{theorem}

These properties---quasi-concavity, symmetry behavior and monotonicity---hold for both \RPD{} and classical projection depth,
arising from the structural features of the outlyingness $O_v(\cdot; P_X)$ rather than the specific regularized direction set $\VV_{\beta}$. 
In contrast, some classical properties---such as vanishing at infinity and affine invariance---do not directly extend to the regularized functional setting. 
This is due to the fact that these properties depend on the geometry of the direction set~$\VV_{\beta}$, 
which is explicitly built into the definition of the~\RPD{}. The following theorem summarizes characteristic features of~\RPD{}.

\begin{theorem}[Novel \RPD{} properties]\label{thm: RPDPropNov}
    Let $X \sim P_X \in \P{\HH}$ and assume that $\beta\geq 0$ is chosen so that~$\VV_{\beta}\neq\emptyset$. Then:
    \begin{enumerate}[label=\upshape{(N\arabic*)}]
        
        \item(\textbf{Continuity})\label{N1} If $\beta > 0$, then $D_{\beta}(\cdot; P_X)$ is $(1/\beta)$-Lipschitz continuous and thus uniformly continuous on $\HH$. As a consequence, the $\tau$-upper level set $D_{\beta}^\tau(P_X)$ is closed and complete for all $\tau \geq 0$.

        \item(\textbf{Orthogonal invariance})\label{N2} The \RPD{} is invariant under shifts and unitary transformations. That is, if $\mathcal{T}x = \mathcal{S}x + e$ for all $x \in \HH$, where $\mathcal{S} \colon \HH \to \HH$ is a unitary operator\footnote{A bounded linear operator $\mathcal{S} \colon \HH \to \HH$ is called unitary if and only if $\mathcal{S}\mathcal{S}^*=\mathcal{S}^*\mathcal{S}=\mathcal{I}$, where $\mathcal{I}$ is the identity operator and $\mathcal{S}^*$ is the adjoint operator of $\mathcal{S}$.} and $e \in \HH$ is a fixed vector, then
        \[
        D_{\beta}(\mathcal{T}x; P_{\mathcal{T}X}) = D_{\beta}(x; P_X) \quad\text{for all } x\in\HH.
        \]

        \item(\textbf{Vanishing at infinity})\label{N3} Let $\set{x_k}_{k=1}^\infty \subset \HH\setminus\{0\}$ be a sequence such that $\norm{x_k} \to \infty$ as $k \to \infty$, and assume that  
        \begin{equation}\label{condVanish}
            \liminf_{k \to \infty} \sup_{v \in \VV_{\beta}} \inner{\frac{x_k}{\norm{x_k}}, v} > 0.
        \end{equation}
        \sloppy Then $D_{\beta}(x_k; P_X) \to 0$ as $k \to \infty$. In particular, if $\set{x_k}_{k=1}^\infty \subset \spn{\VV_{\beta}}$ and $\dim \spn{\{x_k\}_{k=1}^\infty} < \infty$, then condition~\eqref{condVanish} is satisfied.
        
        \item(\textbf{Effective domain reduction})\label{N4} Denote by $\widehat{x}$ is the orthogonal projection of $x$ onto $\Clo{\spn{\VV_{\beta}}}$. Then it holds that $D_{\beta}(x; P_X) = D_{\beta}(\widehat{x}; P_X)$. 
    \end{enumerate}
\end{theorem}

With $\beta > 0$, the~\RPD{} is $1/\beta$-Lipschitz continuous~\ref{N1}, which ensures its regular behavior in the argument $x$; 
the larger the regularization parameter $\beta$, the more stable the depth function becomes.

Invariance under translations and unitary transformations~\ref{N2} is another key property of the \RPD{}, 
capturing essential symmetries such as invariance under coordinate permutations.
In contrast, affine invariance---which is standardly expected from multivariate depths---has no relevance in the functional data setting \citep[Section~2.1]{Gijbels_Nagy2017}.

The~\RPD{} satisfies the vanishing at infinity property~\ref{N3} under an additional geometric condition~\eqref{condVanish}.
Specifically, increasing $\beta$ reduces the direction set $\VV_\beta$,
which, according to~\ref{N3}, limits the depth’s ability to detect extreme deviations perpendicular to these directions.
Consequently, the decay of depth for distant observations is not guaranteed in full generality, but still holds within the constrained geometry of $\VV_\beta$.
Our goal is thus to select the regularization parameters to keep $\VV_\beta$ sufficiently large.
Importantly, the failure to satisfy the vanishing at infinity property is common among functional depths;
see~\citet[Table~1]{Gijbels_Nagy2017} for a summary.
Therefore, property~\ref{N3} provides a valuable guarantee. 

Finally, the \RPD{} value of a point $x \in \HH$ depends solely on its projection onto the closed linear span of the direction set $\VV_\beta$~\ref{N4}.
This effective domain reduction implies that the depth is determined entirely by its behavior on this subspace of the ambient space.
Furthermore, this property emphasizes that regularization not only stabilizes the depth functional,
but also induces an implicit dimension reduction aligned with the most informative directions in the data. In other words, $\Clo{\spn{\VV_{\beta}}}$ carries all the information in $D_{\beta}(\cdot; P_X)$, 
since replacing any point by its projection onto this set leaves the depth value unchanged.

For this reason, it is therefore natural to define the \RPD{} median as
\begin{equation} \label{eq: RPD median}
    \theta(P_X)
    = \argmax\set{D_{\beta}(\theta; P_X)\st\theta \in \Clo{\spn{\VV_{\beta}}}}.
\end{equation}
The \RPD{} median always exists, as shown in the following theorem. Note, however, that the maximizer of $D_{\beta}(\cdot; P_X)$ on $\Clo{\spn{\VV_{\beta}}}$ may not be unique; in such cases, the whole set of maximizers 
    \[  M(P_X) = \left\{ x \in \Clo{\spn{\VV_{\beta}}} \colon D_\beta(x; P_X) = \sup_{y \in \HH} D_\beta(y; P_X) \right\} \subset \HH   \]
is regarded as the set of the \RPD{} medians. Note that the set $M(P_X)$ is convex, but may fail to be compact in $\HH$. Consequently, the usual approach of selecting a single representative as the barycenter of $M(P_X)$ does not apply in the functional setting (because the uniform distribution on a non-compact set $M(P_X) \subset \HH$, necessary to define a barycenter of $M_X$, is not necessarily well defined).

\begin{theorem}[\RPD{} median existence]\label{thm: medianExist}
    Let $X \sim P_X \in \P{\HH}$ and assume that $\beta\geq 0$ is chosen so that~$\VV_{\beta}\neq\emptyset$. Then the \RPD{} median set $M(P_X)$ is non-empty and convex in $\HH$.
\end{theorem}

\subsection{Sample \RPD{} and its consistency}\label{sec3.2}

Let $X \sim P_X \in \P{\HH}$ and $\beta>0$. Consider a random sample $\set{X_i}_{i=1}^n$ of size $n$ with $X_i \eqdis X$, and let $\widehat{P}_n$ denote the corresponding empirical distribution $\widehat{P}_n=n^{-1}\sum_{i=1}^n\delta_{X_i}$, where $\delta_x$ is the Dirac measure at $x\in\HH$. 
For $v\in\SS$,
we write
$\widehat{\med}\left[\set{\inner{X_i, v}}_{i=1}^n\right]$ and
$\widehat{\MAD}\left[\set{\inner{X_i, v}}_{i=1}^n\right]$ for indicating the sample median and sample MAD based on projections $\{\inner{X_i,v}\}_{i=1}^n$,
respectively.\footnote{For a sample $\set{Z_i}_{i=1}^n \subset \R$ with order statistics $Z_{(1)} < \cdots < Z_{(n)}$, we define the sample median as $\widehat{\med}\left[ \set{Z_i}_{i=1}^n\right] = Z_{((n+1)/2)}$ if $n$ is odd and $\widehat{\med}\left[ \set{Z_i}_{i=1}^n\right] = (Z_{(n/2)} + Z_{(n/2+1)})/2$ otherwise. Then, the sample MAD is defined by $\widehat{\MAD}\left[\set{Z_i}_{i=1}^n\right] = \widehat{\med}\left[\set{\abs{Z_i - \widehat{\med}\left[\set{Z_j}_{j=1}^n\right]}}_{i=1}^n\right]$.}
We then define the sample~\RPD{} of $x \in \HH$ as
\begin{equation}
    D_\beta(x; \widehat{P}_n)=\inf_{v\in\widehat{\VV}_{\beta,n}}\left(1+O_v(x; \widehat{P}_n)\right)^{-1}
    =\left(1+\sup_{v\in\widehat{\VV}_{\beta,n}}\frac{\abs{\inner{x,v}-\widehat{\med}\left[\set{\inner{X_i, v}}_{i=1}^n\right]}}{\widehat{\MAD}\left[\set{\inner{X_i, v}}_{i=1}^n\right]}\right)^{-1}, \label{eqRPDsample}
\end{equation}
where 
the sample direction set $\widehat{\VV}_{\beta,n}$ is defined by
\begin{equation*}
    \widehat{\VV}_{\beta,n} 
    = \VV_\beta(\widehat{P}_n)
    = \set{v \in \SS \st \widehat{\MAD}\left[\set{\inner{X_i, v}}_{i=1}^n\right] \geq \beta}.
\end{equation*}

The following theorem establishes the strong consistency of the sample~\RPD{}. It is interesting to note that the proof technique used to derive this result appears to be new in the functional depth literature, and the derived uniform consistency results do not require compactness of the set of functions over which convergence is established.
To state our result, we introduce the maximal outlyingness function $\Gamma_x$ at $x \in \HH$ over $\VV_t$ as
\begin{align} \label{eqMaxOut}
    \Gamma_x(t)=\sup_{v\in\VV_{t}}O_v(x; P_X), \quad t \in [0,\infty). 
\end{align}

\begin{theorem}[Consistency of \RPD{}]\label{thm: consistency}
    Suppose that $X \sim P_X \in \P{\HH}$ has contiguous support and $\beta>0$ is chosen so that~$\VV_{\beta}\neq\emptyset$. 
    Furthermore, assume that $\mathcal{F}\subset\HH$ is a bounded set such that the family $\set{\Gamma_x : x\in\mathcal{F}}$ of maximal outlyingness functions \eqref{eqMaxOut} is equicontinuous at $\beta$, i.e.,
        \begin{equation}\label{consAssump}
            \lim_{t\to \beta}\sup_{x\in\mathcal{F}}\abs{\Gamma_x(t)-\Gamma_x(\beta)}=0.
        \end{equation}
    Then, the sample \RPD{} \eqref{eqRPDsample} is uniformly consistent for the population \RPD{} \eqref{eqRPD} over $\mathcal{F}$ as
        \begin{equation*}
            \sup_{x\in\mathcal{F}}\abs{D_{\beta}(x; \widehat{P}_n) - D_{\beta}(x; P_X)} \xasto[n\to\infty] 0.
        \end{equation*}
\end{theorem}

The key condition~\eqref{consAssump} reduces the problem of establishing finite sample consistency of \RPD{} to the task of verifying continuity of the deterministic map $\Gamma_x$. The following theorem shows that under fairly general conditions, one obtains pointwise consistency of the sample~\RPD{}.

\begin{theorem}\label{thm: consistencyPoint}

    Suppose that $X \sim P_X \in \P{\HH}$ has contiguous support and $\beta>0$ is chosen so that~$\VV_{\beta}\neq\emptyset$,
    and let $\mathcal{F}=\{x\}$ for some fixed $x \in \HH$. 
    Furthermore, assume that
    the maximal outlyingness at $x \in \HH$ is achieved over $\VV_\beta^+$, i.e.,
    \begin{equation}\label{pointwiseCons condition}
        \Gamma_x(\beta)=\sup_{v\in\VV_\beta}O_v(x;P_X)=\sup_{v\in\VV_\beta^+}O_v(x;P_X),
    \end{equation}
    where $\VV_\beta^+=\set{v\in\SS\st \MAD[\inner{X,v}]>\beta}$.
    Then, the condition~\eqref{consAssump} is satisfied, 
    implying that
    \begin{equation*}
        \abs{D_{\beta}(x; \widehat{P}_n) - D_{\beta}(x; P_X)} \xasto[n\to\infty] 0.
    \end{equation*}
\end{theorem}

Formula~\eqref{pointwiseCons condition} essentially states that the supremum defining $\Gamma_x(\beta)$ is not attained solely along directions with $\MAD[\inner{X,v}]=\beta$, but can instead be approximated by directions with $\MAD[\inner{X,v}]>\beta$. In other words, the boundary directions do not contribute to the supremum. As we show next, this condition is fairly mild.

For elliptically symmetric distributions with characteristic function in~\eqref{eq: ElliptProj}, we have (after a suitable reparametrization) $\MAD[\inner{X,v}]=\norm{\Sigma^{1/2}v}$; see~\eqref{eq: ElliptMAD}.  
Fix $0<\beta<\sup_{v\in\SS}\norm{\Sigma^{1/2}v}$. 
It is easy to see that the closure of the set $\set{v\in\SS \colon \norm{\Sigma^{1/2}v}>\beta}$ (with respect to the relative topology on $\SS$) coincides with $\set{v\in\SS \colon \norm{\Sigma^{1/2}v}\geq \beta}$.  
Together with the continuity of the mapping $v\mapsto O_v(x; P_X)$ (see part~\ref{V2} of Theorem~\ref{thm: Vprop}), this immediately implies that the condition~\eqref{pointwiseCons condition} holds.
Moreover, under ellipticity, one also obtains uniform consistency on bounded sets, as established in the following theorem.

\begin{theorem}\label{th: elliptConsistency}
Let $X \sim P_X \in \P{\HH}$ be elliptically symmetric with location~$\mu\in\HH$ and scatter operator $\Sigma$, defined via the characteristic function given in Equation~\eqref{eq: ElliptProj}, parametrized such that $\MAD[\inner{X,v}] = \norm{\Sigma^{1/2} v}$ (see~\eqref{eq: ElliptMAD} for details). Here, $\Sigma$ is a self-adjoint, positive definite, and compact operator on $\HH$.
Assume that $0 < \beta < \sqrt{\sigma_1}$, where $\sigma_1$ is the maximal eigenvalue of $\Sigma$.
Then, for any $M > 0$, the condition~\eqref{consAssump} is satisfied with 
$\mathcal{F} = \set{x \in \HH \st \norm{x} \leq M}$,
and therefore, 
the sample \RPD{} \eqref{eqRPDsample} is uniformly consistent for the population \RPD{} \eqref{eqRPD} over $\mathcal{F}$ as
\[
\sup_{\norm{x} \leq M} \abs{D_\beta(x; \widehat{P}_n) - D_\beta(x; P_X) } \xasto[n \to \infty] 0.
\]
\end{theorem}

\subsection{Robustness of induced median}\label{sec3.3}

Robustness analysis is a cornerstone in the study of depth-based functionals.
The median derived from a depth function is considered robust if small contaminations in the data distribution do not cause large changes in its value.
In our analysis, 
we measure robustness using a notion of the breakdown point, which quantifies
the largest proportion of contamination that can be introduced before the estimator becomes arbitrarily distorted. 
Although breakdown points have been extensively studied in Euclidean spaces \cite[e.g.,][]{Donoho1982,Zuo2003}, 
their investigation for functional depth medians has been comparatively limited.

\sloppy To explain the details, let $X \sim P_X \in \P{\HH}$ and $\beta > 0$ be such that
$\VV_\beta(P_X) = \set{v \in \SS \st \MAD[\inner{X,v}] \geq \beta} \neq \emptyset$. 
For $\varepsilon \in [0,1]$ and $Q \in \P{\HH}$, 
we define the $\varepsilon$-contaminated distribution 
$P_{(Q,\varepsilon)} \in \P{\HH}$ by $Q$ as
\begin{equation*}
    P_{(Q, \varepsilon)}(A) 
    = ((1-\varepsilon)P_X + \varepsilon Q)(A) 
    = (1-\varepsilon)P_X(A) + \varepsilon\, Q(A),
    \quad \text{for all Borel sets } A \subseteq \HH.
\end{equation*}
The corresponding contaminated direction set $\VV_\beta(Q,\varepsilon)$ is defined by
\begin{equation*}
    \VV_\beta(Q,\varepsilon)
    = \VV_\beta(P_{(Q, \varepsilon)})
    = \set{v \in \SS \st \MAD[\inner{X',v}] \geq \beta},
    \quad \text{where } X' \sim P_{(Q,\varepsilon)}.
\end{equation*}
To ensure that the contaminated \RPD{} $D_\beta(\cdot; P_{(Q, \varepsilon)})$ and the subsequent median $\theta(P_{(Q, \varepsilon)})$ remains well-defined, 
we only consider the set $\mathcal{Q}(\varepsilon)$ of contaminating distributions that yield a non-empty direction set defined by
\begin{equation*}
    \mathcal{Q}(\varepsilon)
    = \set{Q \in \P{\HH} \st \VV_\beta(Q,\varepsilon) \neq \emptyset}.
\end{equation*}
This is a natural condition, which is assumed throughout our paper. 
We then introduce the following definition 
of the breakdown point of the \RPD{} median $\theta = \theta(P_X)$ from~\eqref{eq: RPD median} tailored to our infinite-dimensional setting:
\begin{equation}\label{eq: BPdef}
    \varepsilon^*(\theta;P_X)
    = \inf\set{\varepsilon \in [0,1] \st 
        \sup_{Q \in \mathcal{Q}(\varepsilon)} 
        \sup_{v \in \VV_\beta(Q,\varepsilon)}
        \abs{\inner{\theta(P_{(Q,\varepsilon)}) - \theta(P_X), v}} = \infty}.\footnote{Since we assume that both $\VV_\beta(P_X)\neq\emptyset$ and $\VV_\beta(Q,\varepsilon)\neq\emptyset$, the corresponding sets of medians $M(P_X)$ and $M(P_{(Q,\varepsilon)})$ are non-empty by Theorem~\ref{thm: medianExist}. One must decide how to choose the representatives $\theta(P_X)$ and $\theta(P_{(Q,\varepsilon)})$; however, our results are independent of this choice.}
\end{equation}

The next theorem shows that the breakdown point $ \varepsilon^*(\theta;P_X)$ in \eqref{eq: BPdef} of the \RPD{} median $\theta(P_X)$ in \eqref{eq: RPD median} is $1/2$, 
which is the highest possible value among reasonable estimators. 
This indicates that the \RPD{} median possesses good robustness properties, which is consistent with the high breakdown point of the classical multivariate projection depth~\citep[Theorem~3.4]{Zuo2003}. 

\begin{theorem}\label{thm: BP}
    Let $P_X \in \P{\HH}$ and $\beta>0$ be such that $\VV_\beta(P_X)\neq\emptyset$. Then $ \varepsilon^*(\theta;P_X) =1/2$.
\end{theorem}

While the previous theorem guarantees exceptional robustness properties for the \RPD{} median, it is important to note that $\varepsilon^*(\theta; P_X)$ is a weaker notion than the usual breakdown point, in which the estimator 
breaks down when the difference $\norm{\theta(P_{(Q,\varepsilon)}) - \theta(P_X)}$ in norm can be made arbitrarily large~\citep[Section~1.4]{Huber_Ronchetti2009}. 
In infinite-dimensional spaces, establishing the classical norm-based breakdown point can be technically challenging, 
and this issue is rarely addressed directly in the literature. 
Our formulation~\eqref{eq: BPdef} of breakdown point captures the idea that the estimator breaks down once it can be pushed 
arbitrarily far in some direction $v \in \VV_\beta(Q,\varepsilon)$. 
This is consistent with the weak topology in the space $\Clo{\spn{\VV_\beta(Q,\varepsilon)}}$ that we consider throughout the paper. 
Further, in Theorem~\ref{thm: BP}, we do not consider the situation in which the median breaks down in the sense that 
$\MAD[\inner{X',v}] < \beta$ for all $v \in \SS$, that is, when the set $\VV_\beta(Q,\varepsilon)$ 
becomes empty. While one could, in principle, extend the breakdown analysis to include such degeneracies, we do not address that here due to space constraints.

\section{An application: Outlier detection} \label{sec4}

The goal of the present paper is to introduce the projection depth for functional data and develop its theoretical foundations. The more practical issues like the computation of $\RPD{}$, the choice of the tuning parameter $\beta$, and a comprehensive set of applications of \RPD{} to statistical inference, will be part of a separate article. Here, we limit ourselves to a few remarks about the computation of \RPD{} and a single selected application to an outlier detection problem, which allows us to compare \RPD{} with the regularized halfspace depth from \citet{YDL25RHD} and other classical functional depths, and showcase some of the strengths of the new proposal.

\subsection{Practical implementation and tuning}

The~\RPD{} in~\eqref{eqRPD} is naturally a conditional version of the classical projection depth \citep{Zuo2003} from~\eqref{eqFPD} applied in a Hilbert space $\HH$. For that reason, some of the elaborate computational procedures for the projection depth are also applicable in the setting of $\HH$-valued data, see \citet{Liu_Zuo2014b}, \citet{Shao_etal2022}, and the references therein.

In our implementation, we followed a simple random approximation of the \RPD{} proposed by \citet[Section~6]{Dyckerhoff2004}, where the infimum in~\eqref{eqRPD} over $v \in \VV_{\beta}$ is replaced by a minimum over a (large number) $M \geq 1$ of directions $v_1, \dots, v_M \in \VV_{\beta}$ sampled independently from a distribution $Q_V$ supported on the set $\VV_{\beta}$. We obtain the random \RPD{}
    \begin{equation} \label{eq: random RPD}
    D_{\beta}^{(M)}(x; P_X) = \min_{m=1,\dots,M} \left(1 + O_{v_m}(x; P_X) \right)^{-1}.
    \end{equation}
A straightforward adaptation of \citet[Proposition~11]{Dyckerhoff2004} then gives that, as long as $Q_V$ does not vanish on $\VV_{\beta}$, we have a strong convergence result $\lim_{M \to \infty} D_{\beta}^{(M)}(x; P_X) = D_{\beta}(x; P_X)$ a.s., for any $x$ and $P_X$ in $\HH$. Due to the space constraints, the complete proof of this result and the technical conditions on $Q_V$ will be detailed in our follow-up study of the practical properties of \RPD{}.

We have implemented the approximated depth~\eqref{eq: random RPD} efficiently in \textsf{C++}, with an interface in the software environment \textsf{R} via package \textsf{RcppArmadillo} \citep{RcppArmadillo}. The resulting package \textsf{RPD} is available online.\footnote{\url{https://github.com/NagyStanislav/RPD}}
Numerical experiments confirm that for data of low-to-intermediate effective dimension (that is, the dimension of the ``signal'' in the functional data, as opposed to the typically infinite-dimensional ``noise''), our approximate procedure gives stable and reliable results (we use $M = 10\,000$ in~\eqref{eq: random RPD} by default).

Recall that the positiveness of $\beta > 0$ ensures the non-degeneracy of \RPD{}, and thus establishes theoretical guarantees on the performance of \RPD{}.
In our implementation, we set $\beta$ to be the $u$-quantile of $\MAD[\inner{X,V}]$,
where $V$ is independent of $X$ and follows a distribution supported in~$\SS$.\footnote{In our numerical study in Section~\ref{sec:4.2}, we discretize all functional data onto a grid of $101$ points in their domain; the random variable $V$ is then taken uniform on the unit sphere in $\R^{101}$.} 
This ensures that $\VV_{\beta}$ is non-empty.
The practical choice of parameter $u$ follows the general strategies outlined by \citet[Section~2.6]{YDL25RHD}: the \RPD{} with smaller quantile level $u \in (0, 1)$ guarantees sensitivity to various functional outliers and is expected to generally perform well. Overall, we recommend a relatively small (but positive) value $u$, which yields good performance in subsequent numerical studies.

\subsection{Proof of concept: Detection of shape outliers} \label{sec:4.2}

The performance of the \RPD{} is demonstrated in a simple task of detecting shape outlying functional data. We consider five competitors in the Hilbert space $\HH$ of square-integrable real-valued functions on $[0,1]$: 
    \begin{itemize}
    \item \RPD{}: the newly proposed regularized projection depth;
    \item \RHD{}: the regularized halfspace depth from \citet{YDL25RHD} with the default choice of the dimension $J$ (for sample size $n = 500$ used in our simulations, the default choice is $J = \left\lfloor n^{1/7.1} \right\rfloor = 2$);
    \item \RHD{}$_6$: the depth \RHD{} with a relatively high and fixed dimension $J = 6$; 
    \item \FD{}: the integrated halfspace depth \citep{Fraiman_Muniz2001, Nagy_etal2016}
        \[
        FD(x; P_X) = \int_0^1 HD(x(t); P_{X( t)}) \dd t,
        \]
    where $HD(u; Q) = \min\left\{ Q((-\infty, u]), Q([u, \infty)) \right\}$ is the one-dimensional halfspace depth of $u \in \R$ w.r.t. $Q \in \P{\R}$, and $P_{X(t)} \in \P{\R}$ is the one-dimensional distribution of $X(t)$, $t \in [0,1]$. The depth $FD$ is known to behave quite similarly to, e.g., the widely used modified band depth \citep{Lopez_Romo2009}, and
    \item \ID{}: the infimal halfspace depth \citep{Mosler2013}
        \[
        ID(x; P_X) = \inf_{t \in [0,1]} HD(x(t); P_{X(t)}),
        \]
    as a representative of the infimal/extremal depths from the literature \citep{Narisetty_Nair2016}.
    \end{itemize}
For the precise definition of \RHD{} we refer to \cite{YDL25RHD}. Using the $(1-u)$-quantile for the \RHD{} (to ensure consistency with the regularization for \RPD{}), the regularized depths \RPD{} and \RHD{} are implemented with three choices of quantile levels $u \in \{0.01, 0.05, 0.1\}$, corresponding to increasingly strict regularization. The functional data is discretized into $101$ equi-spaced points in the domain $[0,1]$.

The data are generated in a $J_{\textrm{true}}$-dimensional basis expansion with $J_{\textrm{true}} = 6$ and the basis of orthogonal polynomials $\phi_1, \dots, \phi_6$ on $[0,1]$ of maximum degree $J_{\textrm{true}} - 1 = 5$ ($\phi_1$ is the constant function of degree $0$, etc.). The clean data are a random sample $X_1, \dots, X_n$ of size $n = 500$ generated from a model $X = \sum_{j=1}^6 c_j \phi_j$ with $c = (c_1, \dots, c_6)\tr$ a centered Gaussian random vector. The covariance matrix $\Sigma \in \R^{6 \times 6}$ of $c$ has a unit diagonal and all off-diagonal terms $0.95$. The sample $X_1, \dots, X_{500}$ is contaminated by independent curves $Y_1, \dots, Y_{50}$ distributed as $Y = \sum_{j=1}^6 d_j \phi_j$, where $d = (d_1, \dots, d_6)\tr$ is also Gaussian, but is centered around $(1, \dots,1)\tr \in \R^6$ and with covariance matrix $\Sigma^{-1}/100$. This configuration ensures that the coefficient vectors $c$ are well separated from most of those in $d$ in $\R^6$. The cloud of vectors $d$ is, however, largely contained in the axes-aligned box where most of the realizations from $c$ are located, violating mainly their covariance structure. This results in $X$ and $Y$ being of different shapes, but this difference is not immediately clear to discern from their plots, see Figure~\ref{fig: outl funs}.

\begin{figure}
    \centering
    \includegraphics[width=0.45\linewidth]{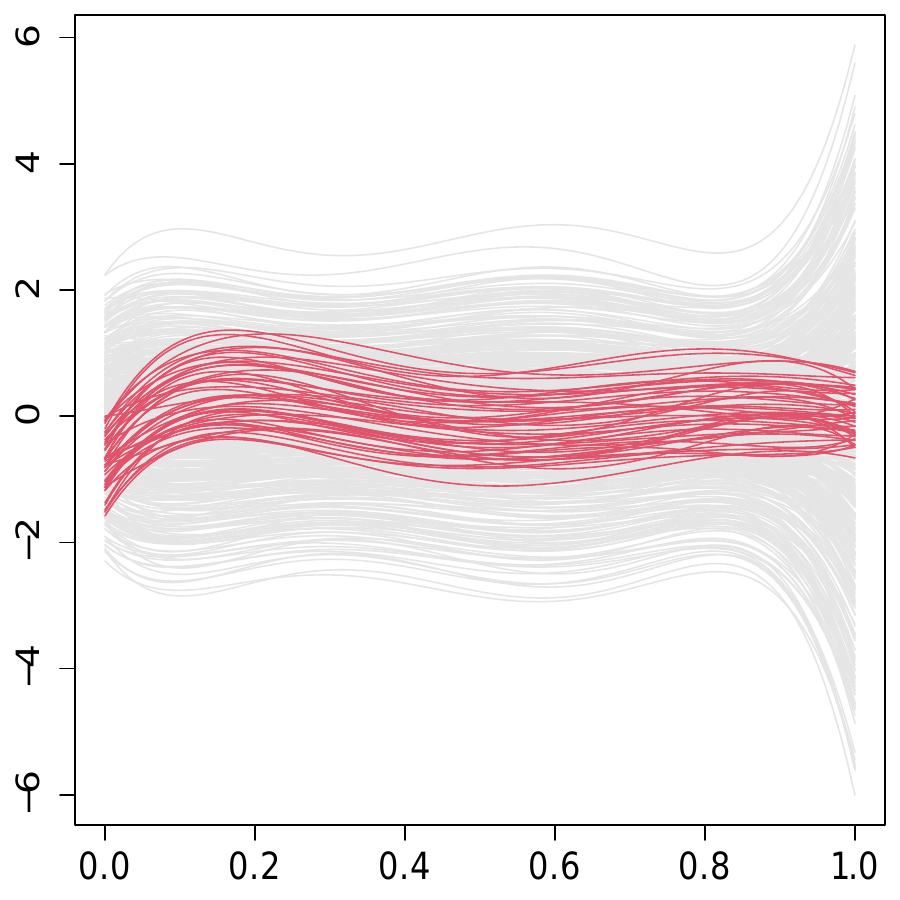}
    \includegraphics[width=0.45\linewidth]{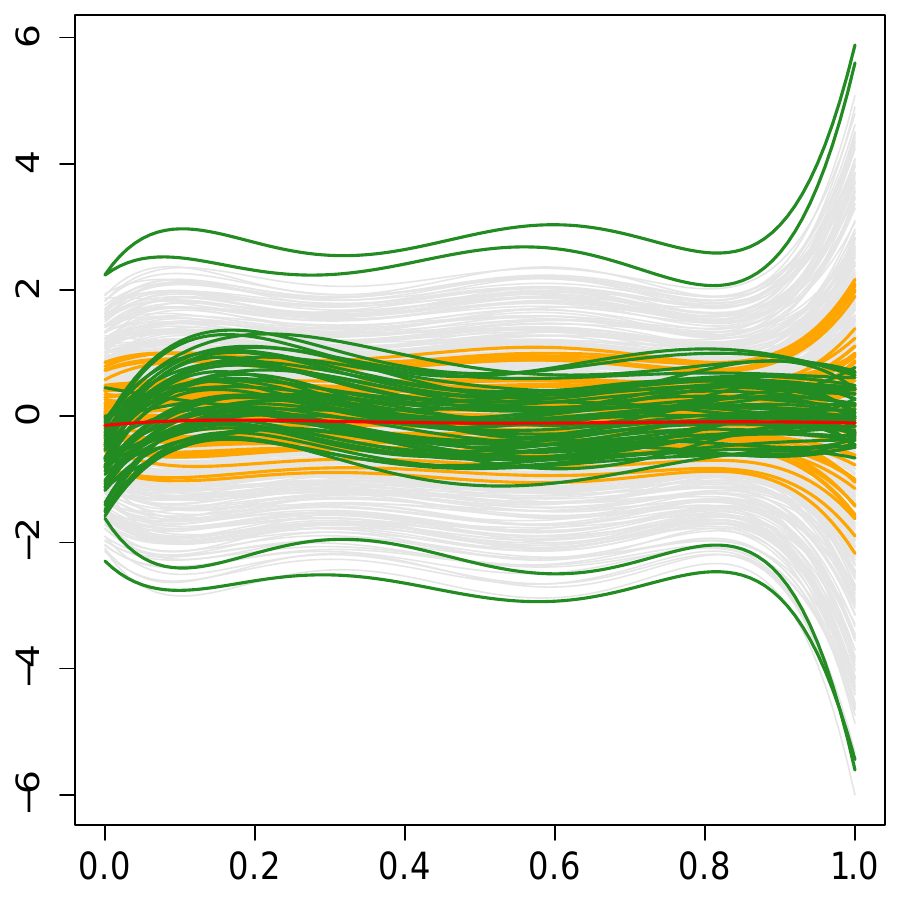}
    \includegraphics[width=0.45\linewidth]{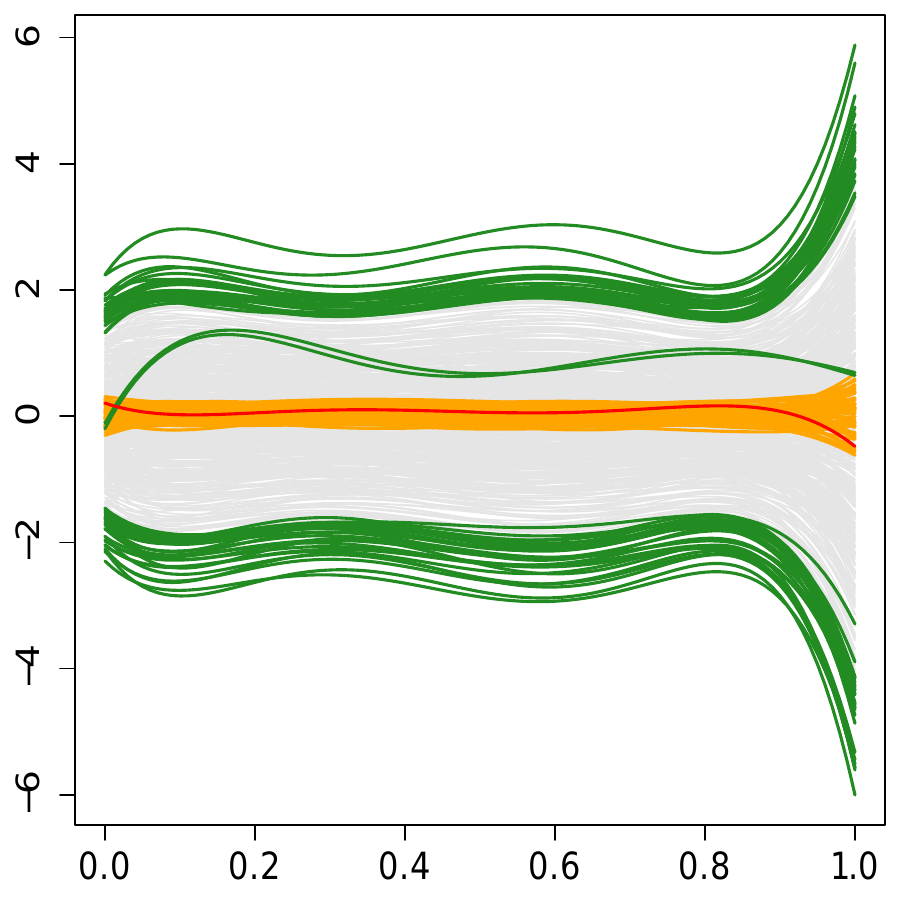}  
    \includegraphics[width=0.45\linewidth]{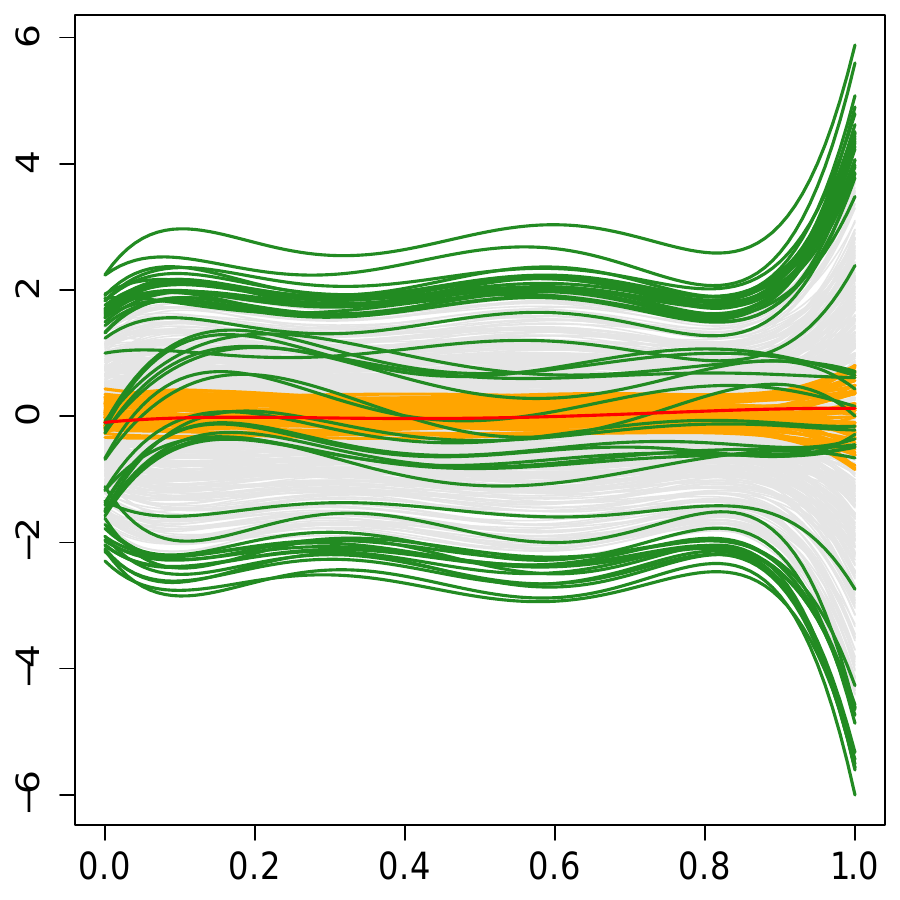}    
    \caption{Outlier detection: A single random sample of clean functional data (gray) of size $n = 500$ with the contaminating sample (red) of size $m = 50$ (top left). The same sample with the sample median curve in red, $10\%$ of the deepest observations in orange, and $10\%$ of the least deep observations in green, using \RPD{} (top right), \RHD{} (bottom left), and \RHD{}$_6$ (bottom right) with $u = 0.01$.}
    \label{fig: outl funs}
\end{figure}

To identify the outliers $Y_i$, $i = 1, \dots, 50$, functional depths of the pooled sample of $500 + 50 = 550$ functions were computed. The functions were ranked according to their depth (rank $1/550 \approx 0$ for the least deep function, and rank $550/550 = 1$ for the deepest curve). The mean rank of the $50$ outliers was evaluated for each functional depth as an indicator of outlier identification performance. This experiment was repeated in $1\,000$ independent Monte Carlo runs, and the means and standard deviations (in brackets) of the resulting mean ranks of outliers are presented in Table~\ref{tab: outl1}. 

\begin{table}[ht]
\centering
\caption{Means and standard deviations (in brackets) of the average rank of outlying curves. The three mean ranks closest to $25/550\approx 0.045$ are highlighted in bold.} 
\label{tab: outl1}
\resizebox{\textwidth}{!}{\begin{tabular}{c|ccc|cc}

$u$ & \RPD & \RHD & \RHD$_6$ & \FD & \ID\\ 
   \hline
$0.01$ & \textbf{0.047} \,(0.001) & 0.469 \,(0.087) & 0.222 \,(0.085) &                  &                  \\ 
  $0.05$ & \textbf{0.048} \,(0.002) & 0.678 \,(0.028) & 0.364 \,(0.173) & 0.703 \,(0.023) & 0.409 \,(0.027) \\ 
  $0.1$  & \textbf{0.058} \,(0.008) & 0.710 \,(0.025) & 0.664 \,(0.037) &                  &                  \\ 

\end{tabular} }
\end{table}

The \RPD{} clearly outperforms all the other depths, achieving the mean rank very close to the minimum possible value $25/550$, 
meaning that nearly all the $50$ outliers were indeed designated among the $50$ least deep data. On the other hand, \RHD{} fails to identify the outliers if the dimension is chosen in a data-driven way, and performs slightly better when the effective dimension of the data is used in \RHD{}$_6$. Still, however, the gap in performance compared with \RPD{} is substantial. 
As expected, due to the pure shape outlyingness of $Y$'s, the integrated depth \FD{} completely fails in this task, 
and based on their rather ``central'' location within the cloud of curves, the \FD{} flags the $Y$'s as very deep in the sample. Slightly better results are obtained using \ID{}, though still quite behind \RPD{}. 

In Figures~\ref{fig: outl funs} and~\ref{fig: outl scatter} are diagnostic plots from one data generation to explain the clear distinction of \RPD{} from \RHD{}. First, we see from Figure~\ref{fig: outl funs} that \RPD{} (top right panel) correctly identifies the shape outliers as curves of low depth.
In contrast, \RHD{} (bottom left panel) attaches more importance to location (horizontal shift) than severe differences in shape, because the dimensionality of the data is reduced too much (the default space into which the functions project is two-dimensional). The depth \RHD{}$_6$, using the information about the true dimensionality of the data, does identify some outliers correctly, but struggles with false positives. The \RPD{} does a good job identifying the shape outliers, especially with low regularization $u$ (though \RPD{} performs quite well across all considered choices of $u$). As a final diagnostic plot, in Figure~\ref{fig: outl scatter}, we see the scatter plot of \RPD{} against \RHD{}$_6$ for each sample curve in a single run of our simulation (with varying~$u$). 
The \RPD{} clearly separates outliers (brightly colored points) from the rest of the data, while for \RHD{}$_6$, even with the correctly identified dimension $J = J_{\mathrm{true}} = 6$ and the highest amount of regularization ($u = 0.01$), many of the non-outlying functions attain depths lower than the outliers.

The limited effectiveness of \RHD{} in functional outlier detection, which at times performs worse even than classical depths such as \FD{} and \ID{}, is attributed to the discrete nature of \RHD{}. This results in tied rankings and restricts its ability to distinguish multiple outliers. From this perspective, the \RPD{} offers substantial practical advantages by enabling outlier detection (and inference) directly via \RPD{} values, without the need for additional tie-breaking/tuning procedures.

\begin{figure}[ht!]
    \centering
    \includegraphics[width=0.45\linewidth]{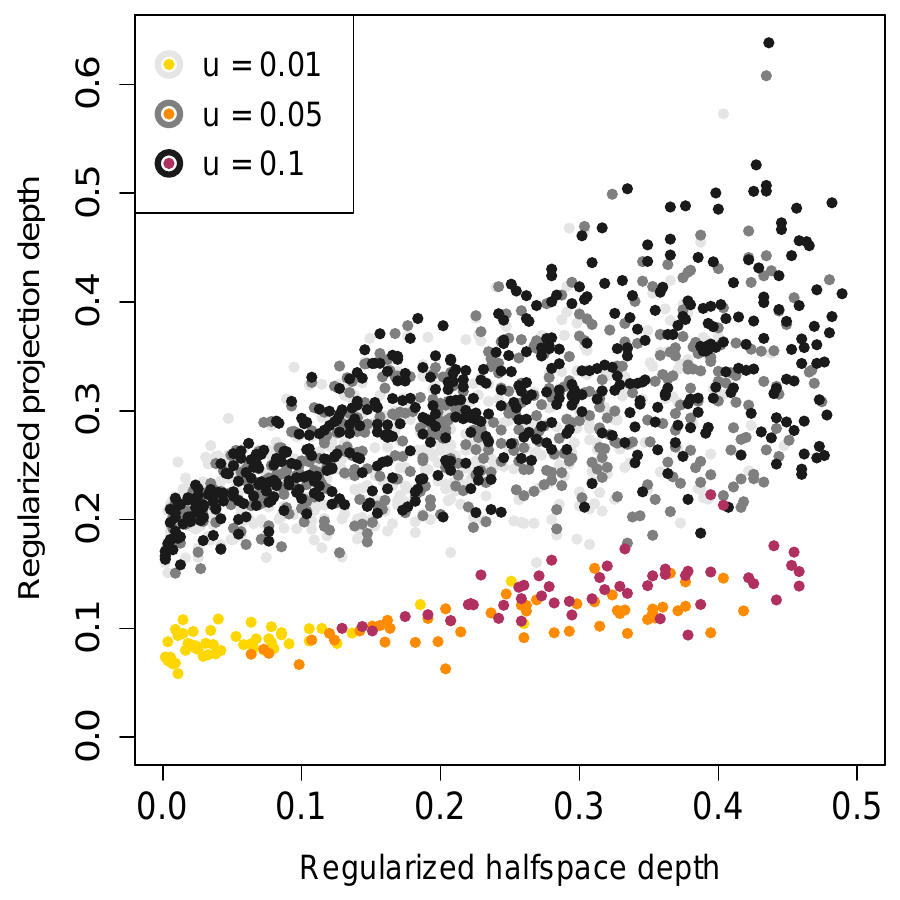}
    \caption{Outlier detection: A scatter plot of \RHD$_6$ (horizontal axis) against \RPD{} (vertical axis) for clean observations (gray) and outliers (gold/orange/violet). Color coding distinguishes the value of the quantile level $u$ (the regularization parameter). \RPD{} correctly flags the outliers as having low depth value, which is not possible using \RHD{}$_6$. This is consistent across the three choices of $u$.}
    \label{fig: outl scatter}
\end{figure}

\section{Conclusion}
The novel~\RPD{} introduced here offers several key advantages over existing functional depth notions.
Thanks to its grounding in Hilbert space and its definition based on projections,
the~\RPD{} is sensitive to shape-related features in functional data.
This makes it particularly effective for tasks like functional outlier detection; see Section~\ref{sec4}.
In contrast, depth functions developed in Banach spaces, such as the integrated~\citep{Fraiman_Muniz2001, Nagy_etal2016} or infimal~\citep{Mosler2013, Narisetty_Nair2016} depth, often struggle with identifying shape outliers.
The~\RPD{} also overcomes some drawbacks of the recently introduced regularized halfspace depth \citep{YDL25RHD}, such as sample depth discreteness and zero values outside the convex hull of the data, which hinder the latter depth in outlier detection by assigning tied values to multiple outliers~\cite[see][Section~3]{YDL25RHD}.
In addition, unlike the Mahalanobis~\citep{BBC20}, $L^2$~\citep{RC24}, or regularized halfspace~\citep{YDL25RHD} depths, the~\RPD{} does not require any moment assumptions on the data,
leading to its natural and strong robustness.
The \RPD{} also satisfies several desirable theoretical properties not shared by other depths; see Section~\ref{sec3}.
Finally, well-known advantages of projection depth in finite dimensions~\citep[e.g.,][]{Zuo2003, Zou06} naturally carry over to the infinite-dimensional~\RPD{} setting.

\appendix

\newpage
\section{Proofs for Sections~1 and~2}

\begin{proof}[\textbf{Proof of Theorem~\ref{thmFPDdege}}]
    Let $\Sigma$ be the positive definite covariance operator of $X$, and assume w.l.o.g. that $\E[X] = 0$. Since projections of a Gaussian random element in $\HH$ are Gaussian, we have for every $v \in \SS$
    \[
        \med[\inner{X, v}] = 0, \quad
        \MAD[\inner{X, v}] = b\, \SD[\inner{X, v}] = b \norm{\Sigma^{1/2}v},
    \]
    where $\SD[\inner{X, v}]$ is the standard deviation of $\inner{X, v}$, and $b = \Phi^{-1}(3/4)$ is the $3/4$-quantile of the standard Gaussian random variable; see~\eqref{eq: ElliptMAD}. The projection depth~\eqref{eqFPD} thus reduces to
    \[
        D(x; P_X) = \inf_{v \in \SS} \left( 1 + b^{-1} \abs{g_v(x)} \right)^{-1},
        \quad \text{where } g_v(x) = \inner{x, v} / \norm{\Sigma^{1/2}v}.
    \]
    Let $\set{\sigma_i}_{i=1}^\infty$ be the non-increasing sequence of positive eigenvalues of $\Sigma$ and let $\set{\phi_i}_{i=1}^\infty$ denote the corresponding eigenvectors. For each $i \in \N$, define $\xi_i = g_{\phi_i}(X_0) = \sigma_i^{-1/2} \inner{X_0, \phi_i}$, so that $\set{\xi_i}_{i=1}^\infty$ is a sequence of i.i.d. standard Gaussian random variables. We claim that
    \begin{align} \label{eqFPDthm1}
        \sup_{i \in \N} \abs{\xi_i} = \infty \quad \text{a.s.}
    \end{align}
    Assuming this, we obtain
    \[
        D(X_0; P_X) \leq \inf_{i \in \N} \left( 1 + b^{-1} \abs{\xi_i} \right)^{-1}
        = \left( 1 + b^{-1} \sup_{i \in \N} \abs{\xi_i} \right)^{-1} = 0 \quad \text{a.s.}
    \]
    It remains to prove~\eqref{eqFPDthm1}. For any $M > 0$ and $J \in \N$,
    \[
        \pr \left( \sup_{i \in \N} \abs{\xi_i} \leq M \right)
        \leq \pr \left( \sup_{1 \leq i \leq J} \abs{\xi_i} \leq M \right)
        = \left\{ 2\Phi(M) - 1 \right\}^J \xto[J\to\infty] 0
    \]
    for $\Phi$ the distribution function of a standard Gaussian random variable. Hence,
    \[
        \pr \left( \sup_{i \in \N} \abs{\xi_i} < \infty \right)
        = \lim_{M \to \infty} \pr \left( \sup_{i \in \N} \abs{\xi_i} \leq M \right) = 0,
    \]
    proving~\eqref{eqFPDthm1} and completing the proof.
\end{proof}

\begin{proof}[\textbf{Proof of Theorem~\ref{thm: nonDeg}}]
    Fix $x \in \HH$. For any $v \in \VV_\beta$ we have
    \begin{align*}
        O_v(x;P_X) 
        &= \frac{\abs{\inner{x, v} - \med[\inner{X,v}]}}{\MAD[\inner{X,v}]}
        \leq \beta^{-1} (\norm{x} + \abs{\med[\inner{X,v}]}).
    \end{align*}
    Since $\inner{X,v}\leq\norm{X}$ and $\abs{\inner{X,v}-\med[\inner{X,v}]}\leq \norm{X}+\abs{\med[\inner{X,v}]}$, we can bound
    \begin{equation}\label{medMADinequalities}
        \med[\inner{X,v}] \leq \med[\norm{X}], \quad \MAD[\inner{X,v}] \leq 2\,\med[\norm{X}].
    \end{equation}
    This implies $\sup_{v \in \VV_\beta} O_v(x;P_X) 
        \leq \beta^{-1} (\norm{x} + \med[\norm{X}])$, or equivalently
    \[
        D_\beta(x;P_X) 
        = \left(1 + \sup_{v\in\VV_\beta} O_v(x;P_X) \right)^{-1}
        \geq \left(1 + \beta^{-1} (\norm{x} + \med[\norm{X}]) \right)^{-1} > 0.
    \]
\end{proof}

\section{Proofs for Section~3.1}

\begin{proof}[\textbf{Proof of Theorem~\ref{thm: Vprop}}]\mbox{}
\begin{enumerate}[label=\upshape{(V\arabic*)}]
    \item This follows from $\norm{-v} = \norm{v}$ and $\MAD[\inner{X, -v}] = \MAD[\inner{X, v}]$.

    \item Let $v_k \xto[k\to\infty] v$ in $\HH$. Then $\inner{X, v_k} \xdistto[k\to\infty] \inner{X, v}$ by the continuous mapping theorem~\citep[Theorem~9.3.7]{Dudley2002}. By~\citet[Theorem~3.6 and Example~3.1]{Huber_Ronchetti2009}, the median is continuous under weak convergence if the limit distribution has connected support, which holds for $\inner{X, v}$ by the assumption on $P_X$. Hence, $\med[\inner{X, v_k}] \xto[k\to\infty] \med[\inner{X, v}]$. Since the mapping $(Q, v) \mapsto \abs{\inner{Y,v} - \med[\inner{Y,v}]}$ for $Y\sim Q$ is continuous at $(P_X,v)$ and the distribution of $X$ has connected support, the same argument applies to $\MAD[\inner{X, v_k}] \xto[k\to\infty] \MAD[\inner{X, v}]$.

     \item 
     If $\beta$ is chosen as the $u$-quantile of $\MAD[\inner{X,V}]$ with $u < 1$, then $P_V(\VV_\beta)\geq 1-u>0$, which implies that $\VV_\beta \neq \emptyset$.
\end{enumerate}
\end{proof}

\begin{proof}[\textbf{Proof of Theorem~\ref{thm: RPDPropInh}}]\mbox{}
\begin{enumerate}[label=\upshape{(P\arabic*)}]
    \item For $\tau \leq 0$, trivially $D_{\beta}^\tau(P_X) = \HH$. For $\tau > 0$, we have
    \begin{equation*}
        D_{\beta}^\tau(P_X) 
        = \bigcap_{v \in \VV_{\beta}} \set{x \in \HH \st \abs{\inner{x, v} - \med[\inner{X, v}]} \leq (\tau^{-1} - 1)\MAD[\inner{X, v}]},
    \end{equation*}
    which is an intersection of convex slabs in $\HH$, hence it is a convex set.

    \item It suffices to show $\med[\inner{X, v}] = \inner{\mu, v}$ for all $v \in \VV_{\beta} \subseteq \SS$. Indeed, this implies $D_{\beta}(\mu; P_X) = 1$ and $D_{\beta}(x; P_X) \leq 1$ for $x \neq \mu$.
    If $P_X$ is halfspace symmetric with contiguous support, then $\pr(\inner{X,v} \geq \inner{\mu,v}) \geq 1/2$ and $\pr(\inner{X,v} \leq \inner{\mu,v}) \geq 1/2$, and the connected support implies uniqueness of the median, i.e. $\med[\inner{X,v}] = \inner{\mu,v}$ for all $v\in\SS$. For centrally symmetric $P_X$, we have $\inner{X,v} - \inner{\mu,v} \eqdis \inner{\mu,v} - \inner{X,v}$, so the set of medians is symmetric about $\inner{\mu,v}$, again implying $\med[\inner{X,v}] = \inner{\mu,v}$.

    \item We have
    \begin{equation*}
        D_{\beta}(x; P_X) 
        = D_{\beta}(x - \mu; P_{X - \mu}) 
        = D_{\beta}(\mu - x; P_{\mu - X})
        = D_{\beta}(2\mu - x; P_X),
    \end{equation*}
    using the shift invariance of the median and the central symmetry of $P_X$.

    \item The first claim follows from Theorem~\ref{thm: nonDeg} and the non-emptiness of $\VV_{\beta}$. Indeed, let $v \in \VV_\beta$. Then, as $\delta \to \infty$, we have $O_v(\delta v; P_X) \to \infty$, and consequently, $D_\beta(\delta v; P_X) \to 0$. The second claim follows directly from the quasi-concavity established in~\ref{P1}.
\end{enumerate}
\end{proof}

To prove Theorem~\ref{thm: RPDPropNov} we need the following lemma.
\begin{lemma}\label{lemma: antipodalVinequality}
    Consider a Hilbert space $\HH$ and a non-empty antipodally symmetric set $\mathcal{U} \subseteq \SS$ (that is, $\mathcal{U} = -\mathcal{U}$) with $\dim(\mathcal{U}) < \infty$. Then, for any sequence $\set{u_k}_{k=1}^\infty \subseteq \spn{\mathcal{U}} \cap \SS$, it holds that $\inf_{k \in \N} \sup_{v \in \mathcal{U}} \inner{u_k, v} > 0$.
\end{lemma}

\begin{proof}
    Let $K = \dim(\spn{\mathcal{U}}) < \infty$, and let $\set{v_\ell}_{\ell=1}^K \subset \SS$ be an orthonormal basis of $\spn{\mathcal{U}}$ (not necessarily $\set{v_\ell}_{\ell=1}^K\subseteq \mathcal{U}$). For any $v \in \mathcal{U}$, write $v = \sum_{\ell=1}^K \delta_\ell(v)\, v_\ell$ with $\delta(v) = (\delta_1(v), \dots, \delta_K(v))\tr \in \R^K$. Then $\inner{\delta(v), \delta(w)} = \inner{v, w}$, in particular $\norm{\delta(v)}_2 = 1$. Here, $\norm{\cdot}_2$ is the usual Euclidean norm on $\R^K$. Since $\mathcal{U}$ spans $\spn{\mathcal{U}}$, the set $\{\delta(v) : v \in \mathcal{U} \} \subset \R^K$ spans $\R^K$. Thus, there exist $w_1, \dots, w_K \in \mathcal{U}$ such that $\{\delta(w_\ell)\}_{\ell=1}^K$ spans $\R^K$. For each $j \in \{ 1, \dots, K \}$, let $e_j \in \R^K$ be the $j$th canonical basis vector, and consider its expansion $e_j = \sum_{\ell=1}^K \alpha_{j,\ell} \delta(w_\ell)$, where $\alpha_{j,\ell} \in \R$.  

    Since $\mathcal{U}$ is symmetric, we have $\sup_{v \in \mathcal{U}} \inner{u_k, v} \geq 0$ for all $k$. Assume for contradiction that $\inf_k \sup_{v \in \mathcal{U}} \inner{u_k, v} = 0$. Choose $\varepsilon > 0$ such that  
    \begin{equation}\label{eq: epsChoice}
        \varepsilon^2 < \left( \sum_{j=1}^K \left( \sum_{\ell=1}^K \alpha_{j,\ell} \right)^2 \right)^{-1}
    \end{equation}
    and find $k \in \N$ such that $\sup_{v \in \mathcal{U}} \inner{u_k, v} \leq \varepsilon$. In particular, we have $\inner{u_k, w_\ell} = \inner{\delta(u_k), \delta(w_\ell)} \leq \varepsilon$. Thus, we obtain  
    \begin{equation*}
        1 = \norm{\delta(u_k)}_2^2 = \sum_{j=1}^K \inner{\delta(u_k), e_j}^2 = \sum_{j=1}^K \left( \sum_{\ell=1}^K \alpha_{j,\ell} \inner{\delta(u_k), \delta(w_\ell)} \right)^2
        \leq \varepsilon^2 \sum_{j=1}^K \left( \sum_{\ell=1}^K \alpha_{j,\ell} \right)^2,
    \end{equation*}
    which contradicts~\eqref{eq: epsChoice}. This completes the proof.
\end{proof}

\begin{proof}[\textbf{Proof of Theorem~\ref{thm: RPDPropNov}}]\mbox{}

\begin{enumerate}[label=\upshape{(N\arabic*)}]
    
    \item The proof follows similarly to \citet[Theorem 2.2]{Zuo2003}. Let $ x, y \in \HH $. Then
    \begin{align*}
        \abs{D_{\beta}(x; P_X) - D_{\beta}(y; P_X)} 
        &= \left| \frac{1}{1 + \Gamma_x(\beta)} - \frac{1}{1 + \Gamma_y(\beta)} \right| \\
        &= \frac{\abs{\Gamma_x(\beta) - \Gamma_y(\beta)}}{(1 + \Gamma_x(\beta))(1 + \Gamma_y(\beta))} 
        \leq \abs{\Gamma_x(\beta) - \Gamma_y(\beta)},
    \end{align*}
    where $\Gamma_z(\beta) = \sup_{v \in \VV_{\beta}} O_v(z; P_X)$ for $z\in\HH$. It remains to bound the difference
    \begin{align*}
        \abs{\Gamma_x(\beta) - \Gamma_y(\beta)}
        &\leq \sup_{v \in \VV_{\beta}} 
        \frac{ 
            \abs{
                \abs{\inner{x, v} - \med[\inner{X, v}]} - 
                \abs{\inner{y, v} - \med[\inner{X, v}]}
            } 
        }{\MAD[\inner{X, v}]} \\
        &\leq \sup_{v \in \VV_{\beta}} 
        \frac{\abs{\inner{x - y, v}}}{\MAD[\inner{X, v}]} 
        \leq \frac{1}{\beta} \norm{x - y}.
    \end{align*}

    \item Using the shift-invariance of the median and shift-invariance of the MAD, for all $v \in \SS$ we get
    \begin{align*}
        O_{\mathcal{S}v}(\mathcal{T}x; P_{\mathcal{T}X})
        &= \frac{\abs{\inner{\mathcal{S}x + e, \mathcal{S}v} - \med[\inner{\mathcal{S}X + e, \mathcal{S}v}]}}{\MAD[\inner{\mathcal{S}X + e, \mathcal{S}v}]} 
        =\frac{\abs{\inner{\mathcal{S}x, \mathcal{S}v} - \med[\inner{\mathcal{S}X, \mathcal{S}v}]}}{\MAD[\inner{\mathcal{S}X, \mathcal{S}v}]} \\
        &= \frac{\abs{\inner{x, v} - \med[\inner{X, v}]}}{\MAD[\inner{X, v}]}
        = O_v(x; P_X).
    \end{align*}

    To conclude, it suffices to show that $\VV_{\beta}(\mathcal{T}X) = \mathcal{S} \VV_{\beta}$, which follows from the fact that for any $v \in \SS$ it holds that $\MAD[\inner{\mathcal{T}X, \mathcal{S}v}] = \MAD[\inner{\mathcal{S}X, \mathcal{S}v}] = \MAD[\inner{X, v}]$.

    \item By inequalities~\eqref{medMADinequalities} and using the antipodal symmetry of $ \VV_{\beta} $, we have
    \begin{equation}
        \sup_{v \in \VV_{\beta}} O_v(x_k; P_X)
        \geq \sup_{v \in \VV_{\beta}} \frac{\abs{\inner{x_k, v}} - \abs{\med[\inner{X, v}]}}{\MAD[\inner{X, v}]}
        \geq \frac{\sup_{v \in \VV_{\beta}} \inner{x_k, v} - \med[\norm{X}]}{2\,\med[\norm{X}]}.\label{eq:1}
    \end{equation}

    By assumption~\eqref{condVanish} and the fact that $\norm{x_k} \xto[k\to\infty] \infty $, we obtain
    \begin{equation*}
        \liminf_{k \to \infty} \sup_{v \in \VV_{\beta}} \inner{x_k, v} 
        \geq \liminf_{k \to \infty} \norm{x_k} \cdot 
        \liminf_{k \to \infty} \sup_{v \in \VV_{\beta}} 
        \inner{\frac{x_k}{\norm{x_k}}, v} = \infty,
    \end{equation*}
    hence $\sup_{v \in \VV_{\beta}} \inner{x_k, v} \to \infty$. This, together with~\eqref{eq:1}, entails $\sup_{v \in \VV_{\beta}} O_v(x_k; P_X) \to \infty$, implying $D_{\beta}(x_k; P_X) \to 0$ as $k \to \infty$. 
    
    Now assume that $\set{x_k}_{k=1}^\infty \subset \spn{\VV_{\beta}}$ and $\dim \spn{\set{x_k}_{k=1}^\infty} < \infty$. Denote by $\mathbb{L}$ the linear span of $\set{x_k}_{k=1}^\infty$. Then, we can bound
    \begin{equation*}
        \liminf_{k \to \infty} \sup_{v \in \VV_{\beta}} \inner{\frac{x_k}{\norm{x_k}}, v}\geq  \inf_{k\in\N} \sup_{v \in \VV_{\beta}} \inner{\frac{x_k}{\norm{x_k}}, v}
        \geq \inf_{k\in\N}\sup_{v \in \VV_{\beta}\cap\,\mathbb{L}} \inner{\frac{x_k}{\norm{x_k}}, v}>0,
    \end{equation*}
    where the strict inequality holds by applying Lemma~\ref{lemma: antipodalVinequality} with $u_k=x_k/\norm{x_k}$ and a non-empty antipodally-symmetric set $\mathcal{U}=\VV_\beta\cap\,\mathbb{L}$, which is finite-dimensional.

    \item Last part follows directly from the fact that $\widehat{x}$ is uniquely determined by the condition $\inner{\widehat{x}, v} = \inner{x, v}$ for all $v \in \Clo{\spn{\VV_{\beta}}}\supseteq\VV_{\beta}$.
\end{enumerate}
\end{proof}

\begin{proof}[\textbf{Proof of Theorem~\ref{thm: medianExist}}]\mbox{}
    Let $\mathbb{L}=\Clo{\spn{\VV_{\beta}}}$.
    For $\tau \in [0,1]$, denote by $G(\tau)$ the intersection of the $\tau$-upper level set~\eqref{eq: UpperLevelSet} with $\mathbb{L}$, that is $G(\tau)=\set{x \in \mathbb{L} \st D_{\beta}(x; P_X) \geq \tau}$. The sets $G(\tau)$ are closed (by part~\ref{N1} of Theorem~\ref{thm: RPDPropNov}) and convex (by part~\ref{P1} of Theorem~\ref{thm: RPDPropNov}). 
    By~\citet[Theorem 3.1.2]{Rudin1991}, they are also weakly closed. Moreover, for any $\tau \in (0,1]$, the set $G(\tau)$ is weakly bounded in $\mathbb{L}$, meaning that for any $v \in \mathbb{L} \cap \SS$ it holds that $\sup_{x \in G(\tau)} \abs{\inner{x,v}} < \infty$. This follows from parts~\ref{N3} and~\ref{N1} of Theorem~\ref{thm: RPDPropNov}. 
    Since $\mathbb{L}$ is a closed subspace of the Hilbert space $\HH$, it is reflexive. 
    Therefore, as the sets $G(\tau)$, $\tau \in (0,1]$, are weakly closed and weakly bounded, the Banach–Alaoglu theorem \citep[Theorem~3.15]{Rudin1991} implies that they are weakly compact.

    Let $\tau_0 = \sup_{x \in \mathbb{L}} D_\beta(x; P_X) \leq 1$.
    By Theorem~\ref{thm: nonDeg}, we have $\tau_0 > 0$.
    To show that $G(\tau_0) \neq \emptyset$, note that $G(\tau_0) = \bigcap_{0 < \tau < \tau_0} G(\tau)$, which is the intersection of a nested family of non-empty weakly compact sets, and is therefore non-empty. Convexity of $G(\tau_0) = M(P_X)$ follows directly from~\ref{P1} in Theorem~\ref{thm: RPDPropInh}.
\end{proof}

\section{Proofs for Section~3.2}

For the proofs in Section~3.2, let $P_X(v)\in\P{\R}$ denote the distribution of $\inner{X,v}$ and $\widehat{P}_n(v)\in\P{\R}$ the distribution obtained by projecting $\widehat{P}_n$ onto the direction $v \in\HH \setminus\{0\}$. Define
$m(v) = \med[\inner{X, v}]$ and $s(v) = \MAD[\inner{X, v}]$, with corresponding estimators
$\widehat{m}(v) = \widehat{\med}\left[\set{\inner{X_i, v}}_{i=1}^n\right]$ and 
$\widehat{s}(v) = \widehat{\MAD}\left[\set{\inner{X_i, v}}_{i=1}^n\right]$. Both $\widehat{m}$ and $\widehat{s}$, of course, depend on $n$, but this will be suppressed in the notation. Note that since the function $t \mapsto 1/(1 + t)$ is $1$-Lipschitz for $t \geq 0$, we can bound
\begin{multline}
    \abs{D_{\beta}(x; \widehat{P}_n) - D_{\beta}(x; P_X)} \leq \abs{\sup_{v \in \widehat{\VV}_{\beta,n}} O_v(x; \widehat{P}_n) - \sup_{v \in \VV_\beta} O_v(x; P_X)} \\
    \leq \underbrace{\abs{\sup_{v \in \widehat{\VV}_{\beta,n}} O_v(x; \widehat{P}_n) - \sup_{v \in \widehat{\VV}_{\beta,n}} O_v(x; P_X)}}_{U_n(x)} 
    + \underbrace{\abs{\sup_{v \in \widehat{\VV}_{\beta,n}} O_v(x; P_X) - \sup_{v \in \VV_\beta} O_v(x; P_X)}}_{W_n(x)}.\label{eqUnWn}
\end{multline}

Before proving Theorem~\ref{thm: consistency}, we state the following auxiliary lemmas.

\begin{lemma}\label{lemma: medMADunif}
    Let $X \sim P_X \in \P{\HH}$. Assume that $P_X$ has contiguous support. Then it holds that $\sup_{v\in\SS}\abs{\widehat{m}(v)-m(v)} \xasto[n\to\infty] 0$ and $\sup_{v\in\SS}\abs{\widehat{s}(v)-s(v)} \xasto[n\to\infty] 0$.
\end{lemma}

\begin{proof}
We first prove that the real mappings
\begin{gather*}
    \psi_1\colon v\mapsto m(v), \quad 
    \psi_2\colon v\mapsto \widehat{m}(v), \quad
    \psi_3\colon v\mapsto s(v), \quad
    \psi_4\colon v\mapsto \widehat{s}(v),
\end{gather*}
defined on $\HH$ are weakly sequentially continuous.
Indeed, let $v_k\xweakto[k\to\infty] v$, 
meaning that $\inner{u,v_k}\xto[k\to\infty]\inner{u,v}$ for all $u\in\HH$.
In particular, $\inner{X(\omega),v_k}\xto[k\to\infty]\inner{X(\omega),v}$ for all $\omega\in\Omega$. Consequently, $\inner{X,v_k}\xdistto[k\to\infty] \inner{X,v}$. 
By~\citet[Theorem~3.6 and Example~3.1]{Huber_Ronchetti2009}, 
the median is continuous under weak convergence provided that the limit distribution has connected support,
which holds for $\inner{X,v}$ by our assumption on $P_X$. 
Hence, $m(v_k)\xto[k\to\infty] m(v)$.
It follows that $\abs{\inner{X,v_k}-m(v_k)} \xdistto[k\to\infty] \abs{\inner{X,v}-m(v)}$, and the same argument implies $s(v_k)\xto[k\to\infty] s(v)$. 
Thus, $\psi_1$ and $\psi_3$ are weakly sequentially continuous.  

Next, because $v_k\xweakto[k\to\infty] v$, 
we also have $\inner{X_i, v_k}\xto[k\to\infty]\inner{X_i, v}$ for all $i\in\{1,\ldots,n\}$. 
Therefore, all order statistics of $\{\inner{X_i,v_k}\}_{i=1}^n$ converge to the corresponding order statistics of $\{\inner{X_i,v}\}_{i=1}^n$ as $k\to\infty$. 
This yields $\widehat{m}(v_k)\xto[k\to\infty]\widehat{m}(v)$ and $\widehat{s}(v_k)\xto[k\to\infty]\widehat{s}(v)$, 
showing that $\psi_2$ and $\psi_4$ are also weakly sequentially continuous.

\medskip
Note that since we assume a separable Hilbert space $\HH$, it follows that the relative weak topology on the unit ball $\BB=\{v\in\HH \colon \norm{v}\leq 1\}$ is metrizable~\citep[Theorem~3.16]{Rudin1991}, since $\BB$ is weakly compact. Consequently, because $\psi_1,\dots,\psi_4$ are weakly sequentially continuous, the restrictions of $\psi_1,\dots,\psi_4$ to the unit ball $\BB$ are also weakly continuous~\citep[p. 395]{Rudin1991}. 

\medskip
Denote by $d_P$ the Prohorov distance on $\P{\HH}$ \citep[Section~11.3]{Dudley2002}. Since we assume that $\HH$ is separable, it follows that $d_P$ metrizes the weak topology of probability measures on $\P{\HH}$.
By~\citet[Lemma~6.1 and Theorem~6.1]{BBTW11}, we have
\begin{equation}\label{eq: unifProh}
    \sup_{v\in\SS} d_P(\widehat{P}_n(v),P_X(v)) \xasto[n\to\infty] 0.
\end{equation}
Restrict ourselves to the set $\mathcal{A}\subseteq \Omega$ with $\pr(\mathcal{A})=1$ on which this convergence holds.

\medskip
We prove the convergence result for median; 
the case of $\MAD$ is analogous. 
The following argument follows~\citet[proof of Theorem~6.2]{BBTW11}. 
Define $a_n=\sup_{v\in\BB}\abs{\widehat{m}(v)-m(v)}$, and we aim to show that $L=\limsup_{n\to\infty}a_n=0$. 
Suppose, for contradiction, that $L>0$.  

Since $\psi_1$ and $\psi_2$ are weakly continuous on $\BB$,
the mapping $v\mapsto\abs{\widehat{m}(v)-m(v)}$ is weakly continuous on $\BB$.
Because $\BB$ is weakly compact,
for each $n\in\N$ there exists $u_n\in\BB$ such that $a_n = \abs{\widehat{m}(u_n)-m(u_n)}$. Assume that $a_n\xto[n\to\infty] L$; otherwise, choose a subsequence satisfying this. 
Similarly, assume $u_n\xweakto[n\to\infty] u\in\BB$; 
if not, pass to a weakly convergent subsequence, 
which exists due to the weak compactness of $\BB$.  
From~\eqref{eq: unifProh}, $d_P(\widehat{P}_n(u_n),P_X(u_n))\xto[n\to\infty] 0$.
Since $u_n\xweakto[n\to\infty] u$, we also have that $d_P(P_X(u_n),P_X(u))\xto[n\to\infty] 0$.
By the triangle inequality, $d_P(\widehat{P}_n(u_n),P_X(u))\xto[n\to\infty] 0$. As before, because $P_X(u)$ has connected support,
\citet[Theorem~3.6 and Example~3.1]{Huber_Ronchetti2009} implies that $\widehat{m}(u_n)\xto[n\to\infty] m(u)$. Hence $a_n\xto[n\to\infty] 0$, contradicting $L>0$. 
This shows that $L=0$. Since $\SS\subset\BB$, this completes the proof.
\end{proof}

The following lemma addresses the first term $U_n(x)$ of~\eqref{eqUnWn}. 

\begin{lemma}\label{lemma: Un}
    Under the assumptions of Theorem~\ref{thm: consistency}, we have $\sup_{\norm{x}\leq M} U_n(x) \to 0$ as $n \to \infty$, for any $M>0$.
\end{lemma}

\begin{proof}
    Note that, by the assumption~\eqref{consAssump}, there exists $\varepsilon > 0$ such that $\VV_{\beta + \varepsilon} \neq \emptyset$. Therefore, by Lemma~\ref{lemma: medMADunif}, for sufficiently large $n$, it holds a.s. that $\widehat{\VV}_{\beta,n} \neq \emptyset$, which we henceforth assume.
    Furthermore, by Lemma~\ref{lemma: medMADunif}, we may restrict our attention to the set $\mathcal{A} \subseteq \Omega$ with $\pr(\mathcal{A}) = 1$, on which
    \begin{equation}\label{eq: medMADunif}
        \sup_{v\in\SS}\abs{\widehat{m}(v)-m(v)} \xrightarrow[n \to \infty]{} 0,
        \quad
        \sup_{v\in\SS}\abs{\widehat{s}(v)-s(v)} \xrightarrow[n \to \infty]{} 0.
    \end{equation}
        
    Consider $n\in\N$ large enough so that $\sup_{v\in\SS}\abs{\widehat{s}(v)-s(v)} \leq \beta/2$. Then, for any $v\in\widehat{\VV}_{\beta,n}$, it follows that $s(v)\geq \widehat{s}(v)-\beta/2 \geq \beta/2$. Moreover, by~\eqref{medMADinequalities},
    for any $v\in\SS$ we have $\abs{m(v)}\leq \med[\norm{X}]$ and $s(v)\leq 2\,\med[\norm{X}]$. 
    Using these bounds, for any $\norm{x}\leq M$ and all sufficiently large $n$, we obtain
    \begin{align*}
        U_n(x)
        &= \abs{\sup_{v \in \widehat{\VV}_{\beta,n}} O_v(x; \widehat{P}_n) - \sup_{v \in \widehat{\VV}_{\beta,n}} O_v(x; P_X)} \leq \sup_{v \in \widehat{\VV}_{\beta,n}} \abs{O_v(x; \widehat{P}_n)-O_v(x; P_X)} \\
        &\leq \sup_{v \in \widehat{\VV}_{\beta,n}}\abs{\frac{\abs{\inner{x,v}-\widehat{m}(v)}}{\widehat{s}(v)} - \frac{\abs{\inner{x,v}-m(v)}}{s(v)}} \\
        &= \sup_{v \in \widehat{\VV}_{\beta,n}} \frac{1}{\widehat{s}(v)\,s(v)} \Big| s(v)\abs{\inner{x,v}-\widehat{m}(v)} - \widehat{s}(v)\abs{\inner{x,v}-m(v)} \Big| \\
        &\leq \sup_{v \in \widehat{\VV}_{\beta,n}}\frac{1}{\widehat{s}(v)\,s(v)} 
        \Big( (\abs{\inner{x,v}}+\abs{m(v)})\abs{s(v)-\widehat{s}(v)} + s(v)\abs{m(v)-\widehat{m}(v)} \Big) \\
        &\leq \frac{2}{\beta^2}\Big( (M+\med[\norm{X}]) \sup_{v \in \SS}\abs{\widehat{s}(v)-s(v)} 
        + 2\,\med[\norm{X}] \sup_{v \in \SS}\abs{\widehat{m}(v)-m(v)} \Big).
    \end{align*}
    By~\eqref{eq: medMADunif}, the right-hand side converges to $0$ as $n\to\infty$. Since this holds for every $\omega\in\mathcal{A}$ with $\pr(\mathcal{A})=1$, the claim follows.
\end{proof}

Next, we address the second term $W_n(x)$ of~\eqref{eqUnWn}.
Recall that $\Gamma_x(t)=\sup_{v\in\VV_{t}}O_v(x; P_X)$.
\begin{lemma}\label{lemma: Wn}
    Let $X \sim P_X \in \P{\HH}$ have contiguous support. Consider $\beta>0$ and $\mathcal{F}\subset\HH$. Assume that $\set{\Gamma_x\st x\in\mathcal{F}}$ is equicontinuous at $t=\beta$, i.e.,
    \begin{equation*}
        \lim_{t\to \beta}\sup_{x\in\mathcal{F}}\abs{\Gamma_x(t)-\Gamma_x(\beta)}=0.
    \end{equation*}
    Then $\sup_{x\in\mathcal{F}} W_n(x)\xasto[n\to\infty] 0$.
\end{lemma}
\begin{proof}
    Denote $\delta_n = \sup_{v\in\SS}\abs{\widehat{s}(v)-s(v)}$. Note that for all $v\in\SS$ we have 
    \[
        s(v)\geq \beta+\delta_n \;\implies\; \widehat{s}(v)\geq \beta
        \quad\text{and}\quad
        \widehat{s}(v)\geq \beta \;\implies\; s(v)\geq \beta-\delta_n.
    \]
    Consequently,
    \begin{equation*}
        \VV_{\beta+\delta_n}=\set{v\in\SS\st s(v)\geq \beta+\delta_n}\subseteq \widehat{\VV}_{\beta,n} \subseteq \set{v\in\SS\st s(v)\geq \beta-\delta_n}=\VV_{\beta-\delta_n}.
    \end{equation*}
    It follows that 
    \begin{equation*}
        \Gamma_x(\beta+\delta_n)=\sup_{v\in\VV_{\beta+\delta_n}}O_v(x; P_X)\leq
        \sup_{v\in\widehat{\VV}_{\beta,n}}O_v(x; P_X)\leq
        \sup_{v\in\VV_{\beta-\delta_n}}O_v(x; P_X)=\Gamma_x(\beta-\delta_n).
    \end{equation*}
    Hence, we can bound 
    \[
    \sup_{x\in\mathcal{F}} W_n(x)\leq \sup_{x\in\mathcal{F}}\abs{\Gamma_x(\beta-\delta_n)-\Gamma_x(\beta+\delta_n)}.
    \]
    Since $\delta_n = \sup_{v\in\SS}\abs{\widehat{s}(v)-s(v)}\xasto[n\to\infty] 0$, this converges almost surely to $0$ by the assumption.
\end{proof}

We are now ready to prove Theorem~\ref{thm: consistency}.
\begin{proof}[\textbf{Proof of Theorem~\ref{thm: consistency}}]\mbox{}
    Recall that, by~\eqref{eqUnWn}, we have
    \begin{equation*}
            \sup_{x\in\mathcal{F}}\abs{D_{\beta}(x; \widehat{P}_n) - D_{\beta}(x; P_X)} \leq \sup_{x\in\mathcal{F}}U_n(x) + \sup_{x\in\mathcal{F}} W_n(x).
    \end{equation*}
    Since we assume that $\mathcal{F}\subset \HH$ is bounded, there is $M>0$ such that $\mathcal{F}\subseteq\set{x\in\HH\st \norm{x}\leq M}$. By Lemma~\ref{lemma: Un}, it follows that $\sup_{x\in\mathcal{F}}U_n(x) \leq \sup_{\norm{x}\leq M}U_n(x) \xasto[n\to\infty] 0$. Lemma~\ref{lemma: Wn}, together with the equicontinuity assumption, implies that also $\sup_{x\in\mathcal{F}} W_n(x) \xasto[n\to\infty] 0$, concluding the proof.
\end{proof}

\begin{proof}[\textbf{Proof of Theorem~\ref{thm: consistencyPoint}}]\mbox{}
First, let $\beta_k \xto[k\to\infty] \beta$ from the left, that is, $0 < \beta_k \leq \beta$ for each $k$. 
Then 
\begin{equation} \label{eq: VV order}
    \VV_\beta \subseteq \VV_{\beta_k} \quad \text{for all $k$},  
\end{equation} 
and consequently $\Gamma_x(\beta) 
    = \sup_{v \in \VV_\beta} O_v(x; P_X) 
    \leq \sup_{v \in \VV_{\beta_k}} O_v(x; P_X) 
    = \Gamma_x(\beta_k)$. Suppose, for the sake of contradiction, that $\Gamma_x(\beta_k)$ does not converge to $\Gamma_x(\beta)$. 
Then, after passing to a subsequence if necessary, 
there exists $\varepsilon > 0$ such that $\Gamma_x(\beta) + \varepsilon \leq \Gamma_x(\beta_k)$ for all $k$. By the definition of the supremum and by~\eqref{eq: VV order}, for each $k$ there exists $v_k \in \VV_{\beta_k}$ such that 
\begin{equation}    \label{eq: Gamma inequality}
    \Gamma_x(\beta) + \varepsilon/2 \leq O_{v_k}(x; P_X) \leq \Gamma_x(\beta_k).
\end{equation}
Consider the sequence $\{v_k\}_{k=1}^\infty \subset \SS$. 
Since the unit ball $\BB=\set{v\in\HH\st \norm{v}\leq 1} \supset \SS$ is weakly compact in $\HH$ \citep[Theorem~3.15]{Rudin1991}, there exists $v_0 \in \BB$ such that (again, passing to a subsequence if necessary) $v_k \xweakto[k\to\infty] v_0$.
As shown in the proof of Lemma~\ref{lemma: medMADunif}, 
both $v \mapsto m(v)$ and $v \mapsto s(v)$ are weakly continuous on $\BB$. 
In particular, $m(v_k) \xto[k\to\infty] m(v_0)$ and $s(v_k) \xto[k\to\infty] s(v_0)$. 
Note that we only know $v_0 \in \BB$. 
However, since $s(v_k)\geq \beta_k$ and $\beta_k \xto[k\to\infty] \beta$, 
it follows that $s(v_0)\geq \beta>0$, hence $v_0\neq 0$. 
Therefore, 
\begin{equation}\label{eq: limit}
    O_{v_k}(x; P_X) 
    = \frac{\abs{\inner{x, v_k} - m(v_k)}}{s(v_k)} 
    \xto[k\to\infty] 
    \frac{\abs{\inner{x, v_0} - m(v_0)}}{s(v_0)} 
    = O_{v_0}(x; P_X).
\end{equation}
Since $v_0 \in \BB \setminus \{0\}$, set $\tilde{v}_0 = v_0/\|v_0\|$.
Because $s(v_0)\geq \beta$, we have $s[\tilde{v}_0] = s(v_0)/\|v_0\| \geq \beta/\|v_0\| \geq \beta$, so that $\tilde{v}_0 \in \VV_\beta$. 
Moreover, since both numerator and denominator scale linearly, $O_{\tilde{v}_0}(x; P_X) = O_{v_0}(x; P_X)$. Combining this with the limit in~\eqref{eq: limit} and the inequality in~\eqref{eq: Gamma inequality}, we obtain
\[
   \Gamma_x(\beta) + \varepsilon/2 
   \leq O_{v_k}(x; P_X) 
   \xto[k\to\infty] O_{v_0}(x; P_X) 
   = O_{\tilde{v}_0}(x; P_X) 
   \leq \sup_{v \in \VV_\beta} O_{v}(x; P_X) 
   = \Gamma_x(\beta),
\]
a contradiction. 
Hence, $\Gamma_x$ is left continuous at $\beta$.

Now, let $\beta_k \xto[k\to\infty] \beta$ from the right.
Note that $\bigcup_{k\in\N}\VV_{\beta_k}=\VV_\beta^+$.
Therefore
\begin{equation*}
    \lim_{k\to\infty}\Gamma_x(\beta_k)=\lim_{k\to\infty}\sup_{v\in\VV_{\beta_k}}O_v(x; P_X)=\sup_{v\in\VV_\beta^+}O_v(x; P_X)=\Gamma_x(\beta)
\end{equation*}
by assumption~\eqref{pointwiseCons condition}. This shows that $\Gamma_x$ is also right continuous at $\beta$.
\end{proof}

\begin{proof}[\textbf{Proof of Theorem~\ref{th: elliptConsistency}}]\mbox{}
    First, note that the mapping $v\mapsto O_v(x; P_X)$ is Lipschitz continuous, uniformly over $\mathcal{F}$ on $\VV_{\beta/2}$. Indeed, observe that $v\mapsto\MAD[\inner{X,v}]=\norm{\Sigma^{1/2}v}$ is $\norm{\Sigma^{1/2}}$-Lipschitz, where $\norm{\Sigma^{1/2}}$ is the operator norm of $\Sigma^{1/2}$. Consequently, for any $v,w\in\VV_{\beta/2}$ we obtain 
    \begin{align*}
        &\abs{O_v(x; P_X)-O_w(x; P_X)}=\abs{\frac{\abs{\inner{x,v}-\inner{\mu,v}}}{\norm{\Sigma^{1/2}v}}-\frac{\abs{\inner{x,w}-\inner{\mu,w}}}{\norm{\Sigma^{1/2}w}}}\\
        &=\abs{\frac{\norm{\Sigma^{1/2}w}\abs{\inner{x,v}-\inner{\mu,v}}-\norm{\Sigma^{1/2}v}\abs{\inner{x,w}-\inner{\mu,w}}}{\norm{\Sigma^{1/2}v}\norm{\Sigma^{1/2}w}}}\\
        &\leq \frac{4}{\beta^2}\left(\abs{\norm{\Sigma^{1/2}w} - \norm{\Sigma^{1/2}v}}\abs{\inner{x,v}-\inner{\mu,v}}+\norm{\Sigma^{1/2}v}\bigl(\abs{\inner{x,v}-\inner{x,w}}+\abs{\inner{\mu,v}-\inner{\mu,w}}\bigr)\right)\\
        &\leq \frac{4}{\beta^2}\left(\norm{\Sigma^{1/2}}\norm{v-w}(M+\norm{\mu})+\norm{\Sigma^{1/2}}\norm{v-w}(M+\norm{\mu})\right)\\
        &=\norm{v-w}\,\frac{8}{\beta^2}\norm{\Sigma^{1/2}}(M+\norm{\mu})=C\norm{v-w},
    \end{align*}
    where $C>0$ is a constant independent of $v$ and $w$. 
    
    Next, consider the Hausdorff distance 
    \[
    d_H(\VV_t,\VV_\beta)=\inf\set{\varepsilon>0\st \VV_\beta\subseteq \VV_t^\varepsilon\wedge \VV_t\subseteq \VV_\beta^\varepsilon}, 
    \] 
    where $\VV_s^\varepsilon = \bigcup_{x\in \VV_s}\set{y\in\HH\st \norm{x-y}\leq \varepsilon}$. Note that for an arbitrary fixed $x\in\mathcal{F}$ and any $t$ with $\beta/2\leq t<\sqrt{\sigma_1}$, we have 
    \begin{equation}\label{eq: Hausdorff bound}
        \abs{\Gamma_x(t)-\Gamma_x(\beta)}=\abs{\sup_{v\in\VV_t}O_v(x; P_X)-\sup_{v\in\VV_\beta}O_v(x; P_X)}\leq C\, d_H(\VV_t,\VV_\beta).
    \end{equation}
    Indeed, for any $v\in\VV_\beta$ there exists $w\in\VV_t$ such that $\norm{v-w}\leq d_H(\VV_t,\VV_\beta)$. Therefore, 
    \begin{equation*}
        O_v(x; P_X)\leq O_w(x; P_X)+C\,d_H(\VV_t,\VV_\beta)\leq \Gamma_x(t)+C\,d_H(\VV_t,\VV_\beta),
    \end{equation*}
    and by taking the supremum over $v\in\VV_\beta$, we obtain
    $\Gamma_x(\beta)-\Gamma_x(t)\leq C\,d_H(\VV_t,\VV_\beta)$. Similarly, one obtains $\Gamma_x(t)-\Gamma_x(\beta)\leq C\,d_H(\VV_t,\VV_\beta)$, which proves~\eqref{eq: Hausdorff bound}. Since the upper bound in~\eqref{eq: Hausdorff bound} is independent of $x$, we also obtain
    \begin{align}
        \sup_{x\in\mathcal{F}}\abs{\Gamma_x(t)-\Gamma_x(\beta)}
        \leq C\, d_H(\VV_t,\VV_\beta). \label{eqGammaBound}
    \end{align}
    We claim that
    \begin{align} \label{eqHausDistBound}
        d_H(\VV_t,\VV_\beta)
        \leq 
        2 \max\set{1 - \sqrt{\frac{\sigma_1 - t^2}{\sigma_1 - \beta^2}}, 1 - \sqrt{\frac{\sigma_1 - \beta^2}{\sigma_1 - t^2}}}.
    \end{align}
    This implies $d_H(\VV_t,\VV_\beta)\to 0$ as $t\to\beta$,
    which along with the inequality in \eqref{eqGammaBound} yields the desired uniform consistency result.

    It now suffices to derive the upper bound in \eqref{eqHausDistBound}.
    Recall that we have $s(v)=\norm{\Sigma^{1/2}v}=\sqrt{\inner{v, \Sigma v}}$ for any $v\in\SS$.
    Take $\beta<t<\sqrt{\sigma_1}$. Note that $\VV_t\subseteq \VV_\beta$. Take any $v\in\VV_\beta\setminus\VV_t$, i.e., $\beta\leq s(v)<t$. Consider the eigenvector $\phi_1 \in \SS$ of $\Sigma$ associated with the maximal eigenvalue $\sigma_1$. This means that $\Sigma \phi_1 = \sigma_1 \phi_1$, and also $\inner{\phi_1, \Sigma \phi_1} = \sigma_1$. Since these identities are valid for both $\phi_1$ and $-\phi_1$, we may assume without loss of generality that $\inner{v, \phi_1} \geq 0$, and denote $\theta = \arccos\left(\inner{v, \phi_1}\right) \in (0, \pi/2]$. Consider the minor arc on $\SS$ between the directions $v$ and $\phi_1$ that is parameterized by
    \[
    \gamma(r) = \frac{\sin((1-r)\theta) v + \sin(r \theta) \phi_1}{\sin(\theta)} = \alpha_1(r, \theta) v + \alpha_2(r, \theta) \phi_1 \in \SS \quad \mbox{for } r \in [0,1].
    \]
    Take the function $g_\theta(r)$ defined by the right-hand side of
    \[
    \begin{aligned}
        s(\gamma(r))^2=\inner{\gamma(r), \Sigma \gamma(r)} & = \inner{\alpha_1(r, \theta) v + \alpha_2(r, \theta)\phi_1, \Sigma\left( \alpha_1(r, \theta) v + \alpha_2(r, \theta)\phi_1 \right)} \\
        & = \alpha_1^2(r, \theta) \inner{v, \Sigma v} + \alpha_2^2(r, \theta) \inner{\phi_1, \Sigma \phi_1} + 2\, \alpha_1(r, \theta) \alpha_2(r, \theta) \inner{v, \Sigma \phi_1} \\
        & \geq \alpha_1^2(r, \theta) \beta^2 + \alpha_2^2(r, \theta) \sigma_1 + 2 \alpha_1(r, \theta) \alpha_2(r, \theta) \sigma_1 \cos (\theta) \\
        & = \sigma_1 - (\sigma_1 - \beta^2)\left( \frac{\sin((1-r)\theta)}{\sin(\theta)} \right)^2 = g_\theta(r).
        \end{aligned}
        \]
    Clearly, $g_\theta(r)$ is strictly increasing and continuous in $r \in [0,1]$ for all $\theta \in (0, \pi/2]$, with $g_\theta(0) = \beta^2$ and $g_\theta(1) = \sigma_1$. As a function of $\theta$, $g_\theta(r)$ is non-increasing in $\theta \in (0, \pi/2]$ for any $r \in [0,1]$, and thus 
    \[  g_\theta(r) \geq g_{\pi/2}(r) = \sigma_1 - (\sigma_1 - \beta^2)\left(\sin((1-r)\pi/2) \right)^2 \quad \mbox{for all }\theta \in (0, \pi/2] \mbox{ and }r \in [0,1]. 
    \]
    Set 
    \begin{equation} \label{eq: tau0}  
        \tau_0 = \tau_0(t)= g_{\pi/2}^{-1}(t^2) = 1 - \frac{2}{\pi} \arcsin\left( \sqrt{\frac{\sigma_1 - t^2}{\sigma_1 - \beta^2}} \right),
    \end{equation}
    where $g_{\pi/2}^{-1}$ is the inverse function of $g_{\pi/2}$.
    Then we have $t^2 = g_{\pi/2}(\tau_0) \leq g_{\theta}(\tau_0)$ for all $\theta \in (0, \pi/2]$, meaning that also $s(\gamma(\tau_0))^2\geq g_\theta(\tau_0)\geq t^2$, hence $\gamma(\tau_0) \in \VV_t$. As for the distance of $v$ and $\gamma(\tau_0)$, we have
    \begin{align*}
        \norm{v - \gamma(\tau_0)}^2 & = \inner{v - \gamma(\tau_0), v - \gamma(\tau_0)} = 2 - 2\inner{v, \gamma(\tau_0)} = 2 - 2 \inner{v, \alpha_1(\tau_0, \theta)v + \alpha_2(\tau_0, \theta)\phi_1} \\
        & = 2 - 2\left( \alpha_1(\tau_0, \theta) + \alpha_2(\tau_0, \theta) \cos(\theta) \right) = 2 - 2 \cos(\tau_0\, \theta),
    \end{align*}
    which is an increasing function of $\theta \in (0, \pi/2]$. We obtain
    \[
        \begin{aligned}
        \norm{v - \gamma(\tau_0)}^2 & \leq 2 \left( 1 - \cos(\tau_0 \pi/2)\right) = 2 \left( 1 - \sqrt{\frac{\sigma_1 - t^2}{\sigma_1 - \beta^2}} \right),
        \end{aligned}
    \]
    where the last equality follows from~\eqref{eq: tau0}
        \[
        \cos\left( \tau_0 \pi/2 \right) = \cos\left( \frac{\pi}{2} - \arcsin\left(\sqrt{\frac{\sigma_1 - t^2}{\sigma_1 - \beta^2}} \right) \right) = \sin\left( \arcsin\left( \sqrt{\frac{\sigma_1 - t^2}{\sigma_1 - \beta^2}} \right) \right) = \sqrt{\frac{\sigma_1 - t^2}{\sigma_1 - \beta^2}}.
        \]
    We have established that for any $v \in \VV_\beta\supseteq \VV_t$ there exists $w = \gamma(\tau_0) \in \VV_t$ such  that 
    \begin{equation*}
        \norm{v-w}\leq 2 \left( 1 - \sqrt{\frac{\sigma_1 - t^2}{\sigma_1 - \beta^2}} \right).
    \end{equation*}
    Analogously, we could prove that if $0<t<\beta$, then for any $v\in\VV_t\supseteq \VV_\beta$ there exists $w\in\VV_\beta$ such that 
    \begin{equation*}
        \norm{v-w}\leq 2 \left( 1 - \sqrt{\frac{\sigma_1 - \beta^2}{\sigma_1 - t^2}} \right).    
    \end{equation*}
    This derives \eqref{eqHausDistBound} and concludes the proof.
\end{proof}

\section{Proofs for Section~3.3}

\begin{proof}[\textbf{Proof of Theorem~\ref{thm: BP}}]\mbox{}
    We will derive that 
    \begin{align} \label{eqBPproof1}
        \varepsilon^*(\theta;P_X) \geq 1/2.
    \end{align}
    The desired result $\varepsilon^*(\theta;P_X) = 1/2$ then follows 
    by using the argument of \citet[Theorem~2.1]{Lopuhaa1991}
    to prove that the \RPD{} median $\theta(P_X)$ is translation equivariant,
    and hence, that the breakdown point $\varepsilon^*(\theta;P_X)$ in \eqref{eq: BPdef} of $\theta(P_X)$ cannot exceed $1/2$.

    To establish the inequality in \eqref{eqBPproof1}, 
    let $\varepsilon < 1/2$ be given. 
    For any $X' \sim P_{(Q,\varepsilon)}$ with $Q \in \mathcal{Q}(\varepsilon)$, we can bound
    \begin{equation}\label{eq: BPmed}
        \med[\inner{X', v}] \leq \med[\norm{X'}] \leq q,
    \end{equation}
    where $q>0$ denotes the $(1/2 + \varepsilon)$-quantile of $\norm{X}$, a constant that does not depend on~$Q$. 
    The upper bound~\eqref{eq: BPmed} follows from the fact that the $\varepsilon$-contaminated distribution 
    $P_{(Q,\varepsilon)} = (1-\varepsilon)P_X + \varepsilon\, Q$ differs from $P_X$ by at most an $\varepsilon$-fraction of its mass. 
    Therefore, the median of $\norm{X'}$ under $P_{(Q,\varepsilon)}$ cannot exceed the $(1/2 + \varepsilon)$-quantile of $\norm{X}$ under $P_X$.
    Similarly,
    \begin{equation}\label{eq: BPMAD}
        \MAD[\inner{X', v}] \leq 2\,\med[\norm{X'}] \leq 2q.
    \end{equation}

    The inequality in \eqref{eqBPproof1} follows by showing show that
    \[
        \sup_{Q \in \mathcal{Q}(\varepsilon)} 
        \sup_{v \in \VV_\beta(Q,\varepsilon)}
        \abs{\inner{\theta(P_{(Q,\varepsilon)}) - \theta(P_X), v}} < \infty.
    \]
    Suppose, for contradiction, that
    \begin{equation*}
        \sup_{Q \in \mathcal{Q}(\varepsilon)} 
        \sup_{v \in \VV_\beta(Q,\varepsilon)}
        \abs{\inner{\theta(P_{(Q,\varepsilon)}) - \theta(P_X), v}} = \infty.
    \end{equation*}
    Since
    \[
        \sup_{Q \in \mathcal{Q}(\varepsilon)} 
        \sup_{v \in \VV_\beta(Q,\varepsilon)}
        \abs{\inner{\theta(P_{(Q,\varepsilon)}) - \theta(P_X), v}} 
        \leq 
        \sup_{Q \in \mathcal{Q}(\varepsilon)} 
        \sup_{v \in \VV_\beta(Q,\varepsilon)}
        \abs{\inner{\theta(P_{(Q,\varepsilon)}), v}} + \norm{\theta(P_X)},
    \]
    it follows that
    \begin{equation}\label{eq:BPsup}
        \sup_{Q \in \mathcal{Q}(\varepsilon)} 
        \sup_{v \in \VV_\beta(Q,\varepsilon)}
        \abs{\inner{\theta(P_{(Q,\varepsilon)}), v}} = \infty.
    \end{equation}
    Because $\beta > 0$, for any $Q \in \mathcal{Q}(\varepsilon)$ we can bound using~\eqref{eq: BPmed}
    \begin{equation*}
        \sup_{v \in \VV_\beta(Q,\varepsilon)} O_v(0; P_{(Q,\varepsilon)}) 
        = \sup_{v \in \VV_\beta(Q,\varepsilon)} 
          \frac{\abs{\med[\inner{X', v}]}}{\MAD[\inner{X', v}]} 
        \leq q / \beta,
    \end{equation*}
    where $X' \sim P_{(Q,\varepsilon)}$. 
    Since $\theta(P_{(Q,\varepsilon)})$ maximizes $D_{\beta}(\cdot; P_{(Q,\varepsilon)})$ over 
    $\Clo{\spn{\VV_\beta(Q,\varepsilon)}}$, its outlyingness cannot exceed that of 
    $0 \in \Clo{\spn{\VV_\beta(Q,\varepsilon)}}$. Hence,
    \begin{align*}
        q / \beta
        &\geq \sup_{v \in \VV_\beta(Q,\varepsilon)} O_v(0; P_{(Q,\varepsilon)}) 
        \geq \sup_{v \in \VV_\beta(Q,\varepsilon)} O_v(\theta(P_{(Q,\varepsilon)}); P_{(Q,\varepsilon)}) \\
        &= \sup_{v \in \VV_\beta(Q,\varepsilon)} 
           \frac{\abs{\inner{\theta(P_{(Q,\varepsilon)}), v} - \med[\inner{X', v}]}}{\MAD[\inner{X', v}]} \\
        &\geq \sup_{v \in \VV_\beta(Q,\varepsilon)} 
           \frac{\abs{\inner{\theta(P_{(Q,\varepsilon)}), v}} - \abs{\med[\inner{X', v}]}}{2q} 
        \geq \frac{\sup_{v \in \VV_\beta(Q,\varepsilon)} \abs{\inner{\theta(P_{(Q,\varepsilon)}), v}} - q}{2q},
    \end{align*}
    where we used bounds~\eqref{eq: BPmed} and~\eqref{eq: BPMAD}.
    Taking the supremum over all $Q \in \mathcal{Q}(\varepsilon)$ and using~\eqref{eq:BPsup}, we obtain
    \[
        q / \beta \geq 
        \frac{\sup_{Q\in \mathcal{Q}(\varepsilon)} \sup_{v \in \VV_\beta(Q,\varepsilon)} \abs{\inner{\theta(P_{(Q,\varepsilon)}), v}} - q}{2q} 
        = \infty,
    \]
    which is a contradiction. 
    This shows that $ \varepsilon^*(\theta;P_X)\geq 1/2$ in \eqref{eqBPproof1},
    and completes the proof.
\end{proof}

%
%
%

\bibliographystyle{elsarticle-harv} 
\bibliography{allpapers.bib}
\end{document}